\documentclass[twoside,11pt]{article}
\usepackage[hmargin=1in,vmargin=1in]{geometry}
\usepackage{jeffe,graphicx}
\usepackage[utf8x]{inputenc}
\usepackage{microtype}
\usepackage{cite}
\usepackage{enumerate}

\usepackage[charter]{mathdesign}
\usepackage[scaled=.96,osf]{XCharter}
\usepackage{sourcesanspro,inconsolata,eucal}
\SetMathAlphabet{\mathsf}{bold}{\encodingdefault}{\sfdefault}{b}{\updefault}
\SetMathAlphabet{\mathtt}{bold}{\encodingdefault}{\ttdefault}{b}{\updefault}
\SetMathAlphabet{\mathsf}{normal}{\encodingdefault}{\sfdefault}{\mddefault}{\updefault}
\SetMathAlphabet{\mathtt}{normal}{\encodingdefault}{\ttdefault}{\mddefault}{\updefault}
\usepackage{textcomp}

\newtheorem{lemma}{Lemma}[section]
\newtheorem{theorem}[lemma]{Theorem}
\newtheorem{corollary}[lemma]{Corollary}

\numberwithin{figure}{section}

\def\Sym#1{\texttt{#1}}						
\def\Inverse#1{\={#1}}						

\def\Head{\operatorname{\mathit{head}}}		
\def\Tail{\operatorname{\mathit{tail}}}		
\def\Left{\operatorname{\mathit{left}}}		
\def\Right{\operatorname{\mathit{right}}}	
\def\bdry{\partial\!}

\makeatletter
\newcommand{\smash@bar}[4]{%
  \smash{\rlap{\raisebox{-#3\fontdimen5#10}{$\m@th#2\mkern#4mu\mathchar'26$}}}%
}
\newcommand{\lambdabar}{{\mathchoice
  {\smash@bar\textfont\displaystyle{0.25}{2.75}\lambda}
  {\smash@bar\textfont\textstyle{0.25}{2.75}\lambda}
  {\smash@bar\scriptfont\scriptstyle{0.25}{2.75}\lambda}
  {\smash@bar\scriptscriptfont\scriptscriptstyle{0.25}{2.75}\lambda}
}}
\makeatother

\def\Walk{\omega}
\def\PerturbedWalk{\widetilde\omega}
\def\Radius{\rho}
\def\Relevant#1{\overline{#1}}
\def\Lift#1{\widehat{#1}}
\def\Tiling{\Lift{Q}}
\def\Grammar{\mathcal{G}}

\usepackage{arydshln}
\let\Newpage\relax

\begin{document}

\pagestyle{myheadings}
\markboth{Topologically Trivial Closed Walks in Directed Surface Graphs}
		{Jeff Erickson and Yipu Wang}

\begin{titlepage}

\title{Topologically Trivial Closed Walks in Directed Surface Graphs%
\thanks{Research on this paper was partially supported by NSF grant CCF-1408763.  An extended abstract of this paper will be presented at the 33rd International Symposium on Computational Geometry \cite{ccwx}.}}

\author{\href{http://jeffe.cs.illinois.edu}{Jeff Erickson} and
		\href{http://ywang298.web.engr.illinois.edu/}{Yipu Wang}\\[1ex]
University of Illinois at Urbana-Champaign}

\maketitle

\begin{bigabstract}
Let $G$ be a directed graph with $n$ vertices and $m$ edges, embedded on a surface $S$, possibly with boundary, with first Betti number $\beta$.  We consider the complexity of finding closed directed walks in $G$ that are either contractible (trivial in homotopy) or bounding (trivial in integer homology) in $S$.  Specifically, we describe algorithms to determine whether $G$ contains a simple contractible cycle in $O(n+m)$ time, or a contractible closed walk in $O(n+m)$ time, or a bounding closed walk in $O(\beta (n+m))$ time.  Our algorithms rely on subtle relationships between strong connectivity in $G$ and in the dual graph $G^\star$; our contractible-closed-walk algorithm also relies on a seminal topological result of Hass and Scott.  We also prove that detecting simple bounding cycles is NP-hard.

We also describe three polynomial-time algorithms to compute shortest contractible closed walks, depending on whether the fundamental group of the surface is free, abelian, or hyperbolic.  A key step in our algorithm for hyperbolic surfaces is the construction of a context-free grammar with $O(g^2L^2)$ non-terminals that generates all contractible closed walks of length at most $L$, and only contractible closed walks, in a system of quads of genus $g\ge 2$.  Finally, we show that computing shortest simple contractible cycles, shortest simple bounding cycles, and shortest bounding closed walks are all NP-hard. 
\end{bigabstract}

\bigskip
\begin{rightquote}{0.5}
So they did what any savvy business would do.\\
They hired a consultant.  They brought in a contractor.\\
I'm sorry, not a \textbf{con}tractor---a con\textbf{trac}tor.\\
A man who made words smaller by combining them\\ or apostrophizing them.
\quotee{\href{https://youtu.be/dLECCmKnrys}
	{Gary Gulman, \emph{Conan}, July 13, 2016}}
\end{rightquote}

%
%
%

\thispagestyle{empty}
\setcounter{page}{0}
\end{titlepage}

%


\section{Introduction}

%

A key step in several algorithms for surface graphs is finding a shortest closed walk and/or simple cycle in the input graph with some interesting topological property.  There is a large body of work on finding short interesting walks and cycles in undirected surface graphs, starting with Thomassen's seminal \emph{3-path condition} \cite{t-egnsn-90,mt-gs-01}.  For example, efficient algorithms are known for computing shortest non-contractible and non-separating cycles \cite{schema,k-csnco-06,cm-fsnsn-07,multishort,holiest}, shortest contractible closed walks \cite{ccl-fsncd-16}, simple cycles that are shortest in their own homotopy class \cite{tight}, and shortest closed walks in a \emph{given} homotopy  \cite{octagons} or homology class \cite{homcover}, and for detecting simple cycles that are either contractible, non-contractible, or non-separating \cite{ccl-fctpe-11}.  On the other hand, several related problems are known to be NP-hard, including computing shortest splitting closed walks \cite{splitting}, computing shortest separating cycles \cite{c-fscss-10}, computing shortest closed walks in a given homology class \cite{surfcut}, and deciding whether a surface graph contains a simple separating or splitting cycle \cite{ccl-fctpe-11}.

Directed surface graphs are much less understood, in part because they do not share  convenient properties of undirected graphs, such as Thomassen's 3-path condition \cite{t-egnsn-90,mt-gs-01}, or the assumption that if shortest paths are unique, then two shortest paths cross at most once \cite{wobble}.  The first progress in this direction was a pair of algorithms by Cabello, Colin de Verdière, and Lazarus \cite{ccl-fsncd-10}, which compute shortest non-contractible and non-separating cycles in directed surface graphs, with running times $O(n^2\log n)$ and $O(g^{1/2}n^{3/2}\log n)$.  Erickson and Nayyeri described an algorithm to compute the shortest non-separating cycles in $2^{O(g)}n\log n$ time~\cite{homcover}.  Later Fox \cite{f-sntcd-13} described algorithms to compute shortest non-contractible cycles in $O(\beta^3 n\log n)$ time and shortest non-separating cycles in $O(\beta^2 n\log n)$ time on surfaces with first Betti number $\beta$.  (For all these bounds, the input size $n$ is the total number of vertices, edges, and faces of the input graph.)

\medskip
This paper describes the first algorithms and hardness results for finding topologically \emph{trivial} closed walks in directed surface graphs.  Our results  extend similar results of Cabello, Colin de Verdière, and Lazarus \cite{c-fscss-10, ccl-fctpe-11} for undirected surface graphs; however, our algorithms require several new techniques, both in design and analysis.  (On the other hand, our NP-hardness proofs are actually simpler than the corresponding proofs for undirected graphs!)

We present results for eight different problems, determined by three independent design choices.   First, we consider two types of “trivial” closed walks: \emph{contractible} walks, which can be continuously deformed to a point, and \emph{bounding} walks, which are weighted sums of face boundaries.  Simple bounding cycles are also called \emph{separating} cycles.  (See Section~\ref{S:back} for more detailed definitions.)  Second, like Cabello~\etal~\cite{c-fscss-10, ccl-fsncd-10, ccl-fctpe-11}, we carefully distinguish between closed walks and simple cycles throughout the paper.  Finally, we consider two different goals: deciding whether a given directed graph contains a trivial cycle or closed walk, and finding the shortest trivial cycle or closed walk in a given directed graph (possibly with weighted edges).  Crucially, our algorithms do \emph{not} assume that the faces of the input embedding are open disks.  Our results are summarized in Table \ref{T:results}.

\begin{table}[ht]
\centering\footnotesize\sffamily
\def\arraystretch{1.25}
\begin{tabular}{c:c|c:c}
\textbf{Structure}		& \textbf{Surface} & \textbf{Any} & \textbf{Shortest} \\
\hline
Simple contractible cycle 
		& & $O(n)$	& NP-hard \\
%
\hdashline
Contractible closed walk
		& annulus		& $O(n)$	& $O(n^2\log\log n)$  \\ 
		& torus			& $O(n)$	& $O(n^3\log\log n)$  \\ 
		& with boundary	& $O(n)$	& $O(\beta^5 n^3)$ \\ 
		& other			& $O(n)$	& $O(\beta^6 n^9)$ \\ 
\hline
Simple bounding cycle 
		& & NP-hard & NP-hard  \\ 
%
\hdashline
Bounding closed walk 
		& & $O(\beta n)$	& NP-hard \\
\end{tabular}
\caption{Our results; $\beta$ is the first Betti number of the underlying surface.}
\label{T:results}
\end{table}

In Section \ref{S:ccw}, we describe linear-time algorithms to determine whether a directed surface graph contains a simple contractible cycle or a contractible closed walk, matching similar algorithms for undirected graphs by Cabello \etal~\cite{ccl-fctpe-11}.  Our algorithms are elementary: After removing some obviously useless edges, we report success if and only if some face of the embedding has a (simple) contractible boundary.  However, the proofs of correctness require careful analysis of the dual graphs, and the correctness proof for contractible closed walks relies on a subtle topological lemma of Hass and Scott \cite{hs-ics-85}.  We emphasize that these problems are nontrivial for directed graphs \emph{even if} every face of the input embedding is a disk; see Figure \ref{F:torus-grid}.

\begin{figure}[ht]
\centering
\includegraphics[scale=0.4]{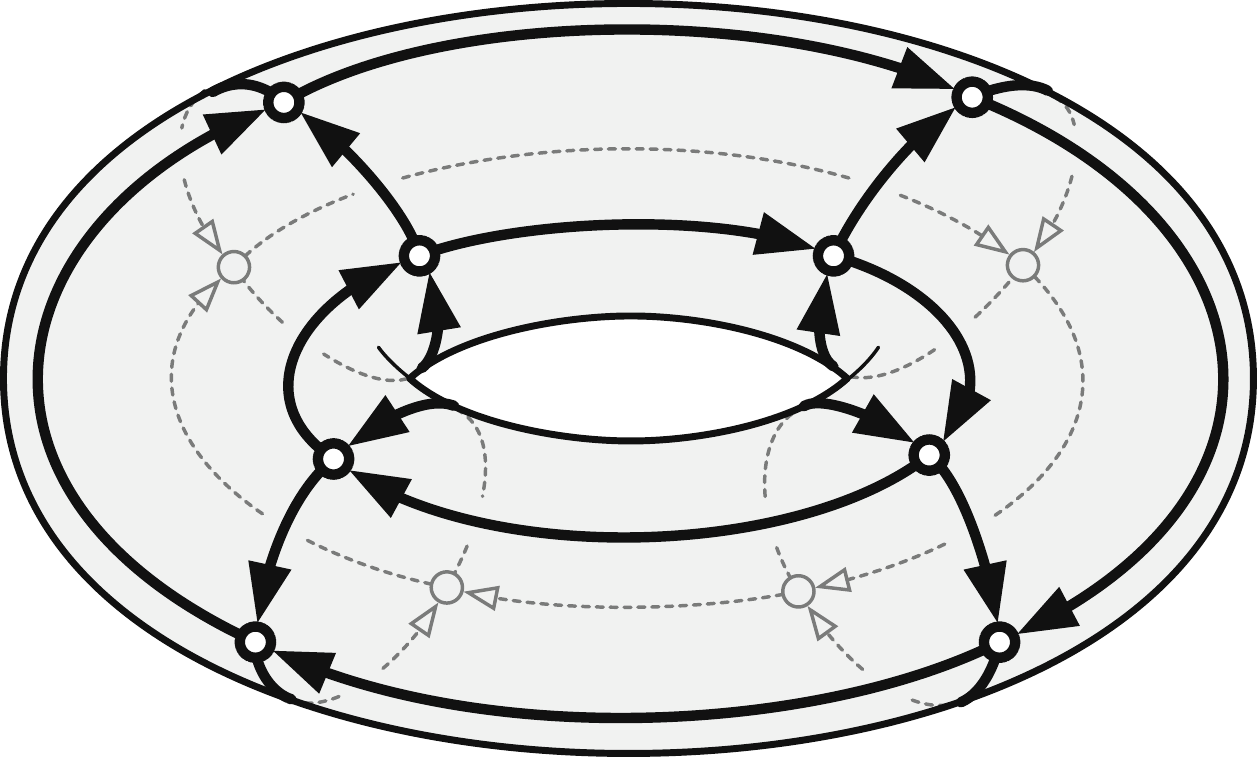}
\caption{A cellularly embedded directed graph with no contractible or bounding closed walks.}
\label{F:torus-grid}
\end{figure}

In Section \ref{S:bcw}, we describe an algorithm to determine whether a directed surface graph contains a \emph{bounding} closed walk in $O(\beta n)$ time.\footnote{This problem is straightforward for undirected surface graphs, even if we forbid bounding walks with spurs.  A connected undirected surface graph supports a spur-free bounding walk if and only if it has either (1) more than one face or (2) more edges than vertices.}  Our algorithm exploits a careful analysis of the interplay between strong connectivity in the input graph $G$ and its dual graph $G^\star$.  With some additional effort, our algorithm can return an explicit description of a bounding closed walk in $O(n^2)$ time if one exists; we prove that this quadratic upper bound is optimal.  This problem can also be reduced to finding zero cycles in periodic (or “dynamic”) graphs \cite{is-scpzs-90,ks-dcdgp-88,cm-spadc-93}.  A periodic graph is a graph whose edges are labeled with integer vectors; a zero cycle is a closed walk whose edge labels sum to the zero vector.  However, all algorithms known for finding zero cycles rely on linear programming, and thus are much more complex and much less efficient than the specialized algorithm we present.  

In Section \ref{S:NP-hard}, we prove that detecting simple bounding cycles, finding shortest simple contractible cycles, and finding shortest bounding closed walks are all NP-hard.  Cabello \cite{c-fscss-10} described an algorithm to compute the shortest simple contractible cycle in an \emph{undirected} surface graph in $O(n^2 \log n)$ time; thus, our reduction for that problem makes essential use of the fact that the input graph is directed.  Cabello and \cite{c-fscss-10} and Cabello \etal~\cite{ccl-fctpe-11} proved that the other two problems are NP-hard in undirected surface graphs.  Our NP-hardness proofs closely follow theirs but are slightly simpler.

Finally, in Sections \ref{S:short-con-easy} and \ref{S:short-con-hyperbolic}, we describe three polynomial-time algorithms to compute shortest contractible closed walks.  Each of our algorithms is designed for a different class of surfaces, depending whether the surface's fundamental group is abelian (the annulus and the torus), free (any surface with boundary), or hyperbolic (everything else).  Our algorithm for the annulus and torus uses a standard covering-space construction, together with a recent algorithm of Mozes \etal\ \cite{mnnw-mdpgo-18}.  For graphs on surfaces with boundary, we exploit the fact that the set of trivial words for any finitely-generated free group is a context-free language \cite{s-cbcgg-76,ms-gtecl-83}; this observation allows us to reduce to a small instance of the \emph{CFG shortest path} problem \cite{y-gmdt-90, bjm-flcpp-00}.  Our algorithm for hyperbolic surfaces also reduces to CFG shortest paths; the main technical hurdle is the construction of an appropriate context-free grammar.  Specifically, for any  integers $g\ge 2$ and $L \ge 1$, we construct a context-free grammar with $O(g^2L^2)$ nonterminals, in Chomsky normal form, that generates all contractible closed walks of length $L$, and only contractible closed walks, in a canonical genus-$g$ surface map called a \emph{system of quads} \cite{lr-hts-12, ew-tcsr-13}.  Our grammar construction exploits classical geometric properties of hyperbolic tilings: linear isoperimetry~\cite{d-tkzf-12,d-pgtt-87} and exponential growth~\cite{m-ncfg-68}.
%

\Newpage
\section{Background}\label{S:back}

\subsubsection*{Directed Graphs}

Let $G$ be an arbitrary directed graph, possibly with loops and parallel edges.  Each edge of $G$ is directed from one endpoint, called its \EMPH{tail}, to the other endpoint, called its \EMPH{head}.  An edge is a \EMPH{loop} if its head and tail coincide.  At the risk of confusing the reader, we sometimes write $\arc{u}{v}$ to denote an edge with tail $u$ and head $v$.

A \EMPH{walk} in $G$ is an alternating sequence of vertices and edges $v_0\arcto v_1 \arcto \cdots \arcto v_\ell$, where $v_i\arcto v_{i+1}$ is an edge in $G$ for each index $i$; this walk is \EMPH{closed} if $v_0 = v_\ell$.  A \EMPH{simple cycle} is a closed walk that visits each vertex at most once.  The \EMPH{concatenation} of two walks $\Walk = v_0 \arcto \cdots \arcto v_{k-1} \arcto v_k$ and $\Walk' = v_k \arcto v_{k+1} \cdots \arcto v_\ell$ is the walk \EMPH{$\Walk\cdot \Walk'$} $:= v_0 \arcto \cdots v_{k-1} \arcto v_k \arcto v_{k+1} \arcto \cdots \arcto v_\ell$.

An \EMPH{edge cut} in a directed graph $G$ is a nonempty subset $X$ of edges such that $G\setminus X$ has two components, one containing the tails of edges in $X$, and other containing the heads of edges in $X$.  A directed graph is \emph{strongly connected} if it contains a directed walk from any vertex to any other vertex, or equivalently, if it contains no edge cuts.

An (integer) \EMPH{circulation} in a directed graph $G$ is a function $\phi\colon E(G)\to\N$ that satisfies a balance constraint $\sum_{\arc{u}{v}} \phi(\arc{u}{v}) = \sum_{\arc{v}{w}} \phi(\arc{v}{w})$ for every vertex $v$.  The \EMPH{support} of a circulation $\phi$ is the subset of edges $e$ such that $\phi(e)>0$.  An \EMPH{Euler tour} of a circulation $\phi$ is a closed walk that traverses each edge $e$ exactly $\phi(e)$ times; such a walk exists if and only if the support of $\phi$ is connected.

\subsubsection*{Surfaces, Embeddings, and Duality}

A \EMPH{surface} is a 2-manifold, possibly with boundary.  A surface is \EMPH{orientable} if it does not contain a Möbius band; we explicitly consider only orientable surfaces in this paper.%
\footnote{As in previous papers \cite{dehn,octagons,lr-hts-12}, all of the problems we consider can be solved for graphs on nonorientable surfaces, with similar running times, by lifting to the oriented double cover.}
A \EMPH{closed curve} on a surface~$S$ is (the image of) a continuous map $\gamma\colon S^1\into S$; a closed curve is \EMPH{simple} if this map is injective.  The \EMPH{boundary $\bdry S$} of $S$ consists of disjoint simple closed curves;the \EMPH{interior} of $S$ is the complement $S^\circ = S\setminus \bdry S$.  The \EMPH{genus} of $S$ is the maximum number of disjoint simple closed curves in $S^\circ$ whose deletion leaves the surface connected.  Up to homeomorphism, there is exactly one orientable surface $S$ with genus $g$ and~$b$ boundary cycles, for any non-negative integers $g$ and $b$.  The \EMPH{first Betti number} of $S$ is either $2g$ if $b=0$, or $2g+b-1$ if $b>0$.

An \EMPH{embedding} of a graph $G$ on a surface $S$ is a continuous map that sends vertices of $G$ to distinct points in~$S^\circ$, and sends edges to interior-disjoint simple paths in $S^\circ$ from their tails to their heads.  In particular, if $G$ contains two anti-parallel edges $\arc{u}{v}$ and its reversal $\arc{v}{u}$, those edges are embedded as interior-disjoint paths.  The embedding of $G$ maps every (simple) closed walk in~$G$ to a (simple) closed curve in $S^\circ$; we usually do not distinguish between a closed walk in $G$ and its image in~$S$.

We explicitly consider graphs with loops and parallel edges; however, without loss of generality, we assume that no loop edge is the boundary of a disk, and no two parallel edges are the boundary of a disk.  (That is, no edge is contractible, and no two edges are homotopic.)  With this assumption, Euler's formula implies that a graph with $n$ vertices on a surface with first Betti number $\beta$ has at most $O(n+\beta)$ edges.

A \EMPH{face} of the embedding is a component of the complement of the image of the graph.  An embedding is \EMPH{cellular} if every face is homeomorphic to an open disk.  \emph{Unlike most previous papers, we explicitly consider non-cellular graph embeddings}; a single face may have disconnected boundary and/or positive genus.  Each directed edge $e$ in a surface graph lies on the boundary of two (possibly equal) faces, called the \EMPH{left shore} and \EMPH{right shore} of $e$.  We sometimes write \EMPH{$f\fenceup f'$} to denote a directed edge whose left shore is $f$ and whose right shore is $f'$.  A \EMPH{boundary face} of the embedding is any face that intersects the boundary of $S$.

A \EMPH{dual walk} is an alternating sequence of faces and edges $f_0\fenceup f_1 \fenceup \cdots \fenceup f_\ell$, where $f_i\fenceup f_{i+1}$ is an edge in~$G$ for each index $i$.  A dual walk is \EMPH{closed} if its initial face $f_0$ and final face $f_\ell$ coincide; a dual walk is \EMPH{simple} if the faces $f_i$ are distinct (except possibly $f_0=f_\ell$).  A simple closed dual walk is called a \EMPH{cocycle}.  Every minimal edge cut in a surface graph is the disjoint union of at most $g+1$ cocycles; in particular, every minimal edge cut in a \emph{planar} directed graph is a cocycle and vice versa.

Every surface embedding of a directed graph $G$ defines a directed \EMPH{dual graph $G^\star$}, with one vertex~$f^\star$ for each face $f$ of $G$, and one edge $e^\star$ for each edge $e$ of~$G$.  Specifically, for any edge $e$ in $G$, the corresponding dual edge satisfies $\Tail(e^\star) = \Left(e)^\star$ and $\Head(e^\star) = \Right(e)^\star$.  The dual graph $G^\star$ can be embedded on the same surface $S$; however, unless the embedding of $G$ is cellular, the embedding of~$G^\star$ is \emph{not} unique.  As a consequence, we will treat $G^\star$ exclusively as an abstract directed graph.  Every dual walk in $G$ corresponds to a walk in $G^\star$; in particular, every cocycle in $G$ corresponds to a cycle in $G^\star$. 

\begin{figure}[ht]
\centering
\includegraphics[scale=0.3]{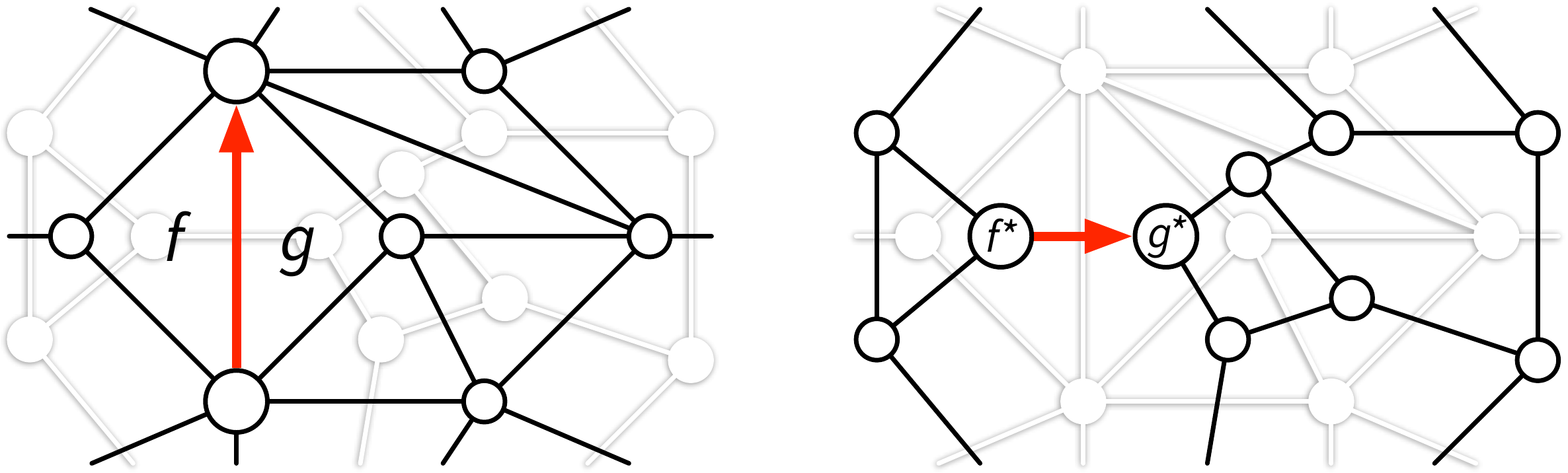}
\caption{Duality in directed surface graphs}
\end{figure}

\subsubsection*{Contractible and Bounding}

Let $\alpha\colon[0,1]\to S$ and $\beta \colon[0,1]\to S$ be two (not necessarily simple) paths in $S$ with the same endpoints.  A \EMPH{homotopy} between $\alpha$ and $\beta$ is a continuous function $h\colon[0,1]^2 \to S$ such that $h(s,0) = \alpha(0) = \beta(0)$ and $h(s,1) = \alpha(1) = \beta(1)$ for all $s$, and $h(0,t) = \alpha(t)$ and $h(1,t) = \beta(t)$ for all $t$.  Two paths are \EMPH{homotopic}, or in the same \EMPH{homotopy class}, if there is a homotopy between them.

A closed curve $\gamma$ in $S$ is \EMPH{contractible} in $S$ if it can be continuously deformed on $S$ to a single point, or more formally, if there is a homotopy from $\gamma$ to a constant function.  A closed walk in a surface graph is contractible if its image under the embedding is a contractible closed curve.  The concatenation of two contractible closed walks is contractible.  A simple closed curve (or a simple cycle in $G$) is contractible in~$S$ if and only if it is the boundary of a disk in $S$ \cite{e-c2mi-66}.

Mirroring classical terminology for curves in the plane~\cite{m-ubiep-1865,a-tikl-28}, an \EMPH{Alexander numbering} for a surface graph~$G$ is a function $\alpha\colon F(G) \to \Z$ that assigns an integer to each face of $G$, such that $\alpha(f) = 0$ for every boundary face $f$.  The \EMPH{boundary $\bdry\alpha$} of an Alexander numbering $\alpha$ is a circulation, defined by setting $\bdry \alpha(e) = \alpha(\Left(e)) - \alpha(\Right(e))$.  A closed walk is \EMPH{bounding} if and only if it is an Euler tour of some boundary circulation.  Equivalently, a closed walk (or its underlying circulation) is bounding if and only if its \emph{integer} homology class is trivial.  (We refer to reader to Hatcher \cite{h-at-02}, Giblin \cite{g-gsh-10}, or Edelsbrunner and Harer \cite{eh-cti-10} for a more through introduction to homology.)  On the sphere or the plane, every closed walk is bounding, and every circulation is a boundary circulation; however, these equivalences do not extend to other surfaces.  (See Figure \ref{F:torus-grid}.)  A simple cycle $\gamma$ in $G$ is bounding in $S$ if and only if $S\setminus\gamma$ is disconnected; simple bounding cycles are usually called \EMPH{separating}.  The concatenation of two bounding closed walks is bounding. 

For example, Figure \ref{F:short-bounding} shows an Alexander numbering for a directed graph on a surface of genus~$2$, which has three vertices, six edges, and three faces.  The boundary circulation of this Alexander numbering has value~$2$ on both loops and value~$1$ on the four edges that are not loops.  The closed walk $a\arcto a\arcto a\arcto b\arcto c\arcto c\arcto c \arcto b\arcto a$ is an Euler tour of this boundary circulation, and thus is a  bounding closed walk.  (This walk is not contractible, and it is actually the shortest bounding walk in this graph.)

\begin{figure}[ht]
\centering\includegraphics[scale=0.5]{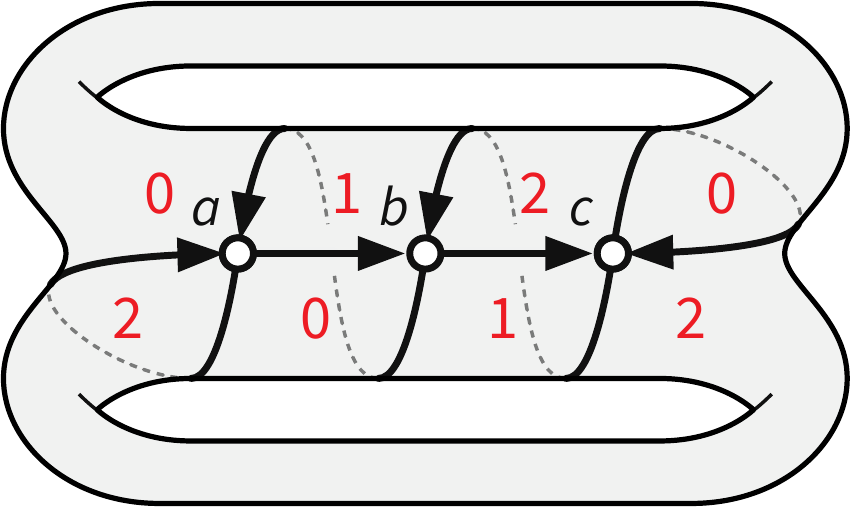}
\caption{An Alexander numbering for a genus-2 surface graph.}
\label{F:short-bounding}
\end{figure} 

Unlike many previous papers, we do not consider homology with coefficients in $\Z_2$, because the problems we consider are easy in that setting.  A directed surface graph contains a closed walk with trivial $\Z_2$-homology if and only if it contains a directed cycle.  Our NP-hardness proofs in Section \ref{S:NP-hard} imply that it is NP-hard to find a simple cycle, the shortest cycle, or the shortest closed walk with trivial $\Z_2$-homology (or with trivial homology over \emph{any} ring).

\Newpage
\section{Contractible Cycles and Walks}
\label{S:ccw}

%
%

\subsection{Simple Cycles}

\begin{lemma}
\label{L:no-cocycle}
Let $G$ be a directed graph embedded on an orientable surface, and let $e$ be a directed edge that lies in a directed cocycle of $G$.  No bounding closed walk (and in particular, no simple contractible cycle) in~$G$ traverses $e$.
\end{lemma}

\begin{proof}
Let $\lambda$ be a directed cocycle in $G$, and let $\Walk$ be a bounding closed walk.  Let $\phi\colon E(G)\to \Z$ be the boundary circulation defined by setting $\phi(e)$ to the number of times that $\Walk$ traverses $e$, and let $\alpha\colon F(G)\to \N$ be an Alexander numbering of $\Walk$ (so that $\phi = \bdry\alpha$).  We immediately have
\[
	\sum_{e\in\lambda} \phi(e)
	~=~ \sum_{e\in\lambda} \big(\alpha(\Right(e)) - \alpha(\Left(e))\big)
	~=~ 0.
\]
Because $\phi(e)\ge 0$ for every edge $e$ in $G$, it follows that $\phi(e)=0$ for every edge $e\in\lambda$.
\end{proof}

In light of this lemma, we can assume without loss of generality that $G$ contains no directed cocycles; in particular, no edge of $G$ has the same face on both sides.

The \EMPH{boundary} of a face $f$, denoted \EMPH{$\bdry f$}, is the set of edges that have $f$ on one side.  The boundary of~$f$ is oriented \EMPH{clockwise} if $f$ is the right shore of any edge in $\bdry f$, and \EMPH{counterclockwise} if~$f$ is the left shore of any edge in $\bdry f$.   Each face in $G$ with counterclockwise boundary appears as a source in the dual graph $G^\star$; each face of $G$ with clockwise boundary appears as a sink in $G^\star$.  We call a face \EMPH{coherent} if its boundary is oriented either clockwise or counterclockwise, and \EMPH{incoherent} otherwise.  See Figure \ref{F:coherent}.
%
  
\begin{figure}[ht]
\centering
\includegraphics[scale=0.4]{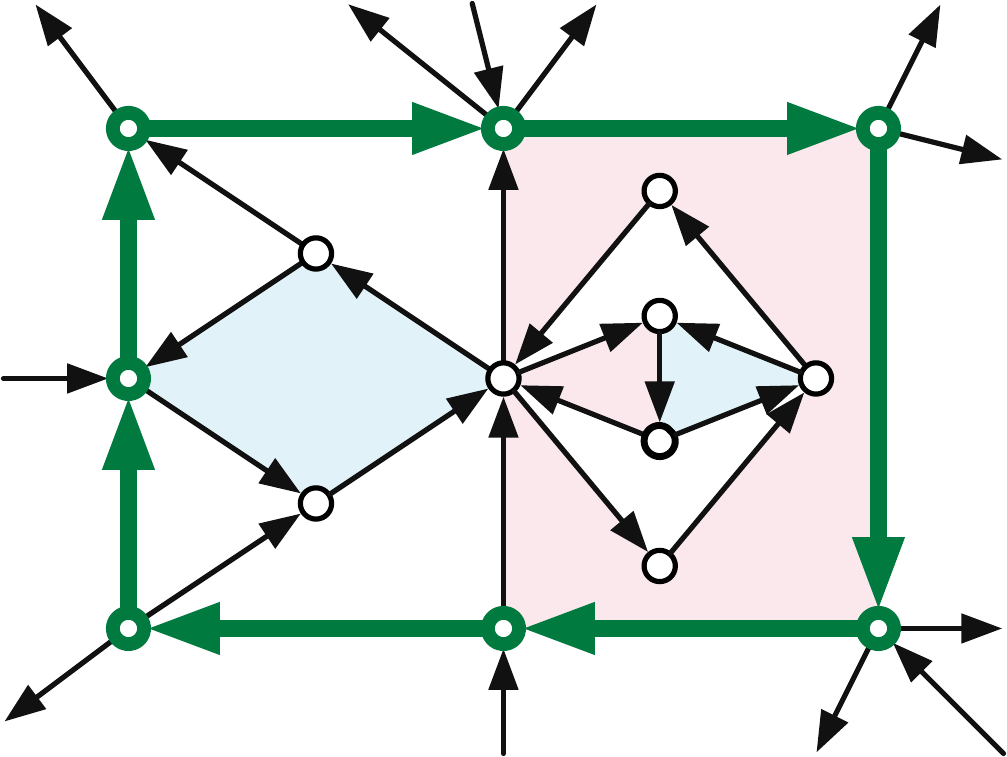}
\caption{A simple clockwise contractible cycle (bold green) in a directed surface graph, enclosing four coherent faces: two counterclockwise (shaded blue) and two clockwise (shaded pink).}
\label{F:coherent}
\end{figure}

\begin{lemma}
\label{L:oneface-cycle}
Let $G$ be a directed graph with no cocycles, embedded on an orientable surface~$S$.  If $G$ contains a simple contractible directed cycle, then $G$ has a face whose boundary is a simple contractible directed cycle.
\end{lemma}

\begin{proof}
Let $\gamma$ be a simple contractible cycle in $G$.  This cycle is the boundary of a closed disk~$D$ \cite{e-c2mi-66}.  Without loss of generality, assume $\gamma$ is oriented clockwise around $D$; that is, the right shore of every edge in $\gamma$ is a face in~$D$.

A dual walk that starts at a face inside $D$ cannot visit the same face more than once, because $G$ has no cocycles ($G^\star$ has no cycles), and cannot cross $\gamma$, because $\gamma$ is oriented clockwise (all its dual edges point into $D$).  Thus, $D$ must contain at least one face $f$ with clockwise boundary (a sink in~$G^\star$).  The closure of~$f$ is a subset of the closed disk~$D$, so it must be homeomorphic to a closed disk with zero or more open disks (“holes”) removed from its interior.

If $\bdry f$ is a simple directed cycle, we are done.  Otherwise, $\bdry f$ is the union of two or more edge-disjoint simple directed cycles; consider the non-simple shaded face in Figure \ref{F:coherent}.  Each of these cycles is contractible, because it lies in $D$.  Any simple cycle in $\bdry f$ that is not the original cycle $\gamma$ (for example, the boundary of any hole in the closure of $f$) encloses strictly fewer faces than $\gamma$.  The lemma now follows immediately by induction on the number of enclosed faces.
\end{proof}

\begin{theorem}
\label{Th:con-cycle}
Given an arbitrary directed graph $G$ embedded on an arbitrary orientable surface, we can determine in $O(n)$ time whether $G$ contains a simple contractible cycle.
\end{theorem}

\begin{proof}
The algorithm proceeds as follows.  First, we compute the strong components of the dual graph~$G^\star$ in $O(n)$ time.  An edge $e$ of $G$ lies in a cocycle if and only if both endpoints of the dual edge $e^\star$ lie in the same strong component of $G^\star$; we remove all such edges from $G$ in $O(n)$ time.  Finally, we examine each face of the remaining subgraph of $G$ by brute force, in $O(n)$ time.  If we discover a face whose interior is a disk and whose boundary is a simple directed cycle, we output \textsc{True}; otherwise, we output \textsc{False}.  Correctness follows directly from Lemmas \ref{L:no-cocycle} and~\ref{L:oneface-cycle}.
\end{proof}

\subsection{Non-simple Closed Walks}

It is easy to construct directed surface graphs that contain contractible closed walks but no simple contractible cycles.  For example, the one-vertex graph in  Figure~\ref{F:torus-weak} has two Eulerian circuits, both of which are contractible, but the only simple cycles consist of single edges, none of which are contractible.

\begin{figure}[ht]
\centering
\includegraphics[scale=0.4]{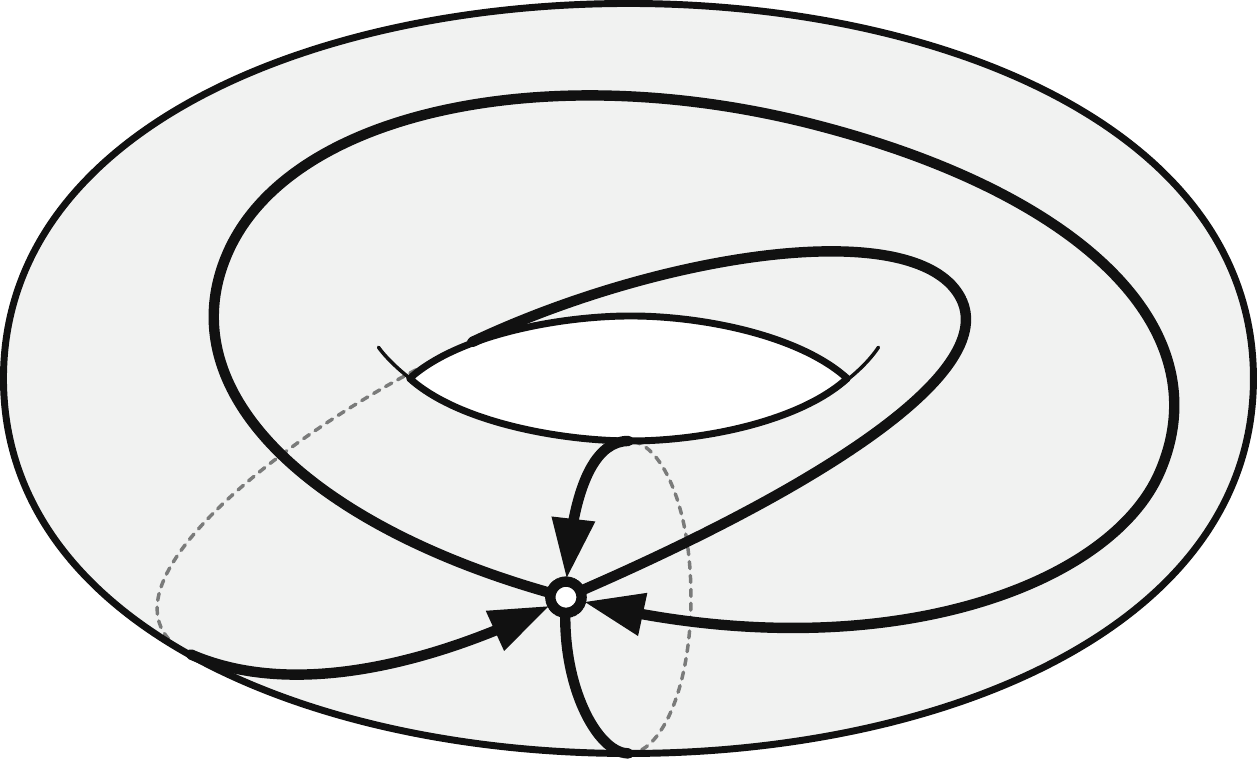}
\caption{A directed surface graph with contractible closed walks but no simple contractible cycles.}
\label{F:torus-weak}
\end{figure}

We detect contractible closed walks using essentially the same algorithm described in the previous section: After removing all cocycle edges, we look for a face whose boundary is coherent and whose interior is an open disk.  However, proving this algorithm correct requires more subtlety.

Fix a graph $G$ embedded on an orientable surface $S$.  Let $\set{U_v \mid v\in V(G)}$ be a collection of disjoint open disks in~$S$, each containing one of the vertices of $G$; we refer to each set $U_v$ as a \emph{vertex bubble}.  Let $\set{U_e \mid e\in E(G)}$ be another collection of disjoint open disks in $S$, each containing the portion of an edge of $G$ outside the vertex bubbles, such that $U_e \cap U_v \ne \varnothing$ if and only if $v$ is one of the endpoints of $e$; we refer to each disk $U_e$ as an \emph{edge bubble}.  Finally, let $U = \bigcup_v U_v \cup \bigcup_e U_e$.  

\begin{figure}[ht]
\centering
\includegraphics[scale=0.4]{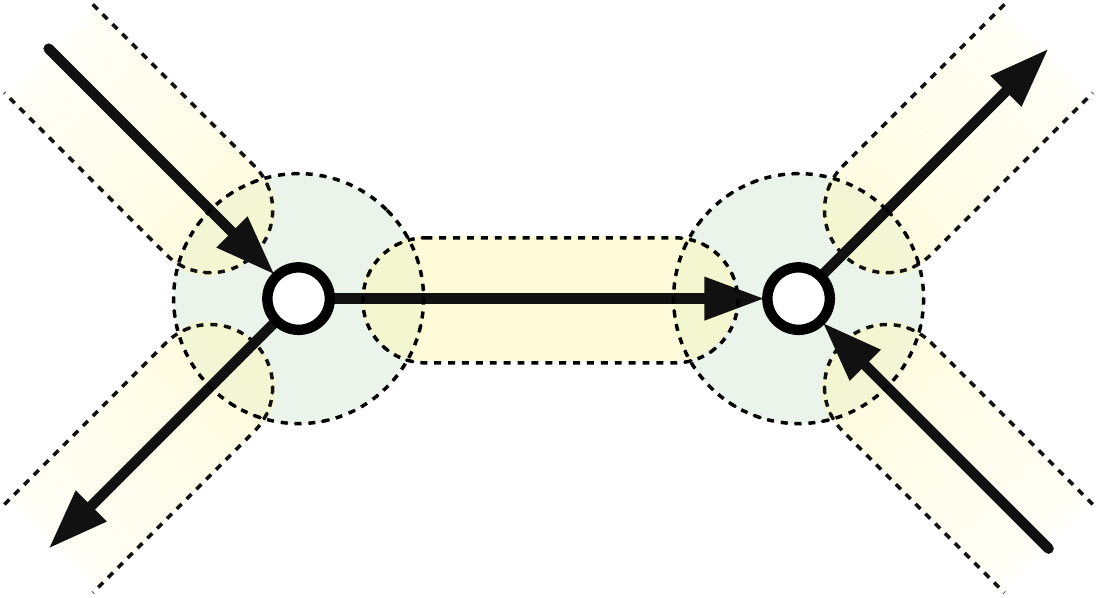}
\caption{Vertex bubbles (green) and edge bubbles (yellow).}
\end{figure}

We say that a simple closed curve $\gamma$ in $U$ \EMPH{follows} a closed walk $\Walk$ in $G$ if $\gamma$ visits the vertex bubbles of~$G$ in exactly the same sequence as the walk $\Walk$ visits the corresponding vertices of~$G$, and therefore $\gamma$ intersects each edge bubble $U_{\arc{u}{v}}$ in a disjoint set of simple paths from $U_u$ to~$U_v$, one for each occurrence of the edge $\arc{u}{v}$ in $\Walk$.  A closed walk~$\Walk$ in $G$ is \EMPH{weakly simple} if $U$ contains a simple closed curve that follows~$\Walk$, or equivalently (because walks in asymmetric directed graphs cannot contain spurs~\cite{weak}) if $\Walk$ is homotopic in~$U$ to a simple closed curve.

\begin{lemma}
\label{L:weak-ccw}
Let $G$ be a directed graph embedded on an orientable surface $S$.  If there is a contractible closed walk in $G$, then there is a weakly simple contractible closed walk in $G$.
\end{lemma}

\begin{proof}
Let $\Walk$ be a contractible closed walk in $G$.  Let $\PerturbedWalk$ be a generic closed curve homotopic to~$\Walk$ inside the graph neighborhood $U$ (and therefore contractible in $S$) that follows $\Walk$ and does not self-intersect inside any edge bubble.  For each edge $\arc{u}{v}$ in $G$, the intersection $\PerturbedWalk \cap U_{\arc{u}{v}}$ consists of pairwise-disjoint simple paths, each directed from $U_u$ to $U_v$.

A seminal result of Hass and Scott \cite[Theorem 2.7]{hs-ics-85} implies that every non-simple contractible closed curve in $U$ contains either a simple contractible closed subpath, which we call a \emph{monogon}, or a pair of simple interior-disjoint subpaths that bound a disk, which we call a \emph{bigon}.  We call a bigon \emph{coherent} if one subpath is homotopic to the reverse of the other, and \emph{incoherent} otherwise.  See Figure \ref{F:smooth}.

\begin{figure}[ht]
\centering
\includegraphics[scale=0.4]{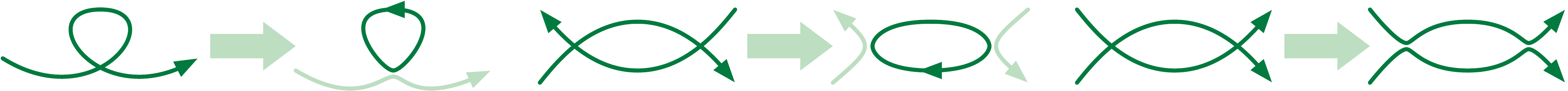}
\caption{Shortening or simplifying a contractible closed curve by smoothing at points of self-intersection.  From left to right: A~monogon, a coherent bigon, and an incoherent bigon.}
\label{F:smooth}
\end{figure}

First, suppose $\PerturbedWalk$ contains two homotopic subpaths~$\alpha$ and $\beta$; these subpaths need not be simple or interior-disjoint.  Let $\PerturbedWalk'$ be the curve obtained from $\PerturbedWalk$ by \emph{smoothing} the common endpoints of $\alpha$ and~$\beta$, as shown on the right of Figure \ref{F:smooth}.  The smoothed curve $\PerturbedWalk'$ is homotopic to $\PerturbedWalk$ via a homotopy that swaps the subpaths $\alpha$ and~$\beta$ \cite{n-csgs-01,dl-cginc-17,untangle}; it follows that~$\PerturbedWalk'$ is contractible.   Thus, by repeatedly smoothing pairs of vertices, we can reduce $\PerturbedWalk$ to a contractible curve~$\PerturbedWalk''$ with no pair of homotopic subpaths, and therefore no incoherent bigons.  (If $\PerturbedWalk$ has no pair of homotopic subpaths, then $\PerturbedWalk'' = \PerturbedWalk$.)

Hass and Scott's theorem implies that either $\PerturbedWalk''$ is simple, or $\PerturbedWalk''$ contains a monogon, or $\PerturbedWalk''$ contains a coherent bigon.  If $\PerturbedWalk''$ contains a monogon, that monogon is itself a simple contractible closed curve.  If $\PerturbedWalk''$ contains a coherent bigon composed of subpaths $\alpha$ and $\beta$, their concatenation $\alpha\cdot\beta$ is a simple contractible closed curve.  In all three cases, we have discovered a simple contractible closed curve $\gamma$ in the graph neighborhood~$U$.  Moreover, for every edge $\arc{u}{v}$ of $G$, the intersection $\gamma \cap U_{\arc{u}{v}}$ is a subset of $\PerturbedWalk\cap U_{\arc{u}{v}}$, and therefore consists of disjoint simple paths directed from $U_u$ to $U_v$.  Thus, $\gamma$ follows a closed walk in~$G$, which is both weakly simple \cite{weak} and contractible.
\end{proof}

\begin{lemma}
\label{L:weak-once}
Let $G$ be a directed graph embedded on an orientable surface $S$.  Every weakly simple \textbf{bounding} closed walk in $G$ traverses each edge of $G$ at most once.
\end{lemma}

\begin{proof}
Let $\Walk$ be a weakly simple bounding closed walk in~$G$.  Let $\gamma$ be a simple closed curve in $U$ that follows $\Walk$, as in the previous proof.  Because $\gamma$ is simple and bounding, the set $S\setminus\gamma$ has exactly two components.  Color the component on the left of $\gamma$ blue, and the component on the right of $\gamma$ red.

Consider a single edge $\arc{u}{v}$.  Each component of $\gamma\cap U_{\arc{u}{v}}$ is a directed path that splits $U_{\arc{u}{v}}$ into a blue region on the left and a red region on the right.  Thus, the components of $\gamma\cap U_{\arc{u}{v}}$ must alternate direction, as shown in Figure \ref{F:stripe}, since otherwise some component of $U_{\arc{u}{v}}\setminus \gamma$ would be both red and blue.  

\begin{figure}[ht]
\centering
\includegraphics[scale=0.4]{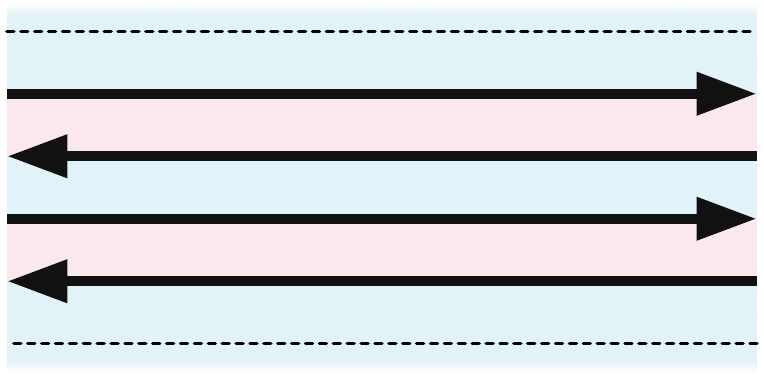}
\caption{A simple contractible cycle in $U$ must alternate directions inside any edge bubble.}
\label{F:stripe}
\end{figure}

By assumption, even if $G$ contains an edge from $v$ to $u$, that edge is embedded disjointly from $\arc{u}{v}$.  Thus, \emph{every} path component of $\PerturbedWalk\cap U_{\arc{u}{v}}$ is directed from $U_u$ to $U_v$.  It follows that $\PerturbedWalk\cap U_{\arc{u}{v}}$ is either empty or a single directed path, which implies that $\Walk$ traverses $\arc{u}{v}$ at most once.
\end{proof}

\begin{corollary}
\label{C:ccw-short}
Let $G$ be a directed graph with $n$ vertices and $m$ non-negatively weighted edges, embedded on an orientable surface, possibly with boundary.  The shortest contractible closed walk in $G$, if such a walk exists, has hop-length at most $m$.
\end{corollary}

\begin{lemma}
\label{L:oneface-walk}
Let $G$ be a directed graph with no cocycles, embedded on an orientable surface.  If $G$ contains a contractible closed walk, then $G$ has a face with coherent boundary whose interior is a disk.
\end{lemma}

\begin{proof}
Let $\Walk$ be a weakly-simple contractible closed walk in $G$ (as guaranteed by Lemma \ref{L:weak-ccw}).  Smoothing~$\Walk$ at its vertices yields a simple contractible closed curve $\gamma$ in $U$ that follows $\Walk$; this closed curve is the boundary of an open disk~$D$ \cite{e-c2mi-66}.  Following terminology for face boundaries, we say that $\Walk$ is oriented \emph{clockwise} if~$D$ lies to the right of $\gamma$ and \emph{counterclockwise} otherwise.  We say that $\Walk$ \emph{encloses} a face~$f$ if the interior region $f\setminus U$ lies inside the disk $D$.

Let $A$ denote the area enclosed by $\Walk$, defined as the interior of the union of the closure of all faces enclosed by $\Walk$.  Area $A$ has genus zero and connected boundary $\Walk$ (by Lemma \ref{L:weak-once}), so each component of~$A$ is an open disk.  If $A$ is disconnected, then the boundary of each component of $A$ is a contractible closed walk that encloses fewer faces than $\Walk$.  Thus, we can  assume without loss of generality that~$A$ is a single open disk.  However, we \emph{cannot} assume that $\Walk$ itself is simple; the closure of $A$ might have holes and/or positive genus.  (See Figures \ref{F:torus-weak} and \ref{F:blob}.)

\begin{figure}[ht]
\centering
\includegraphics[scale=0.4]{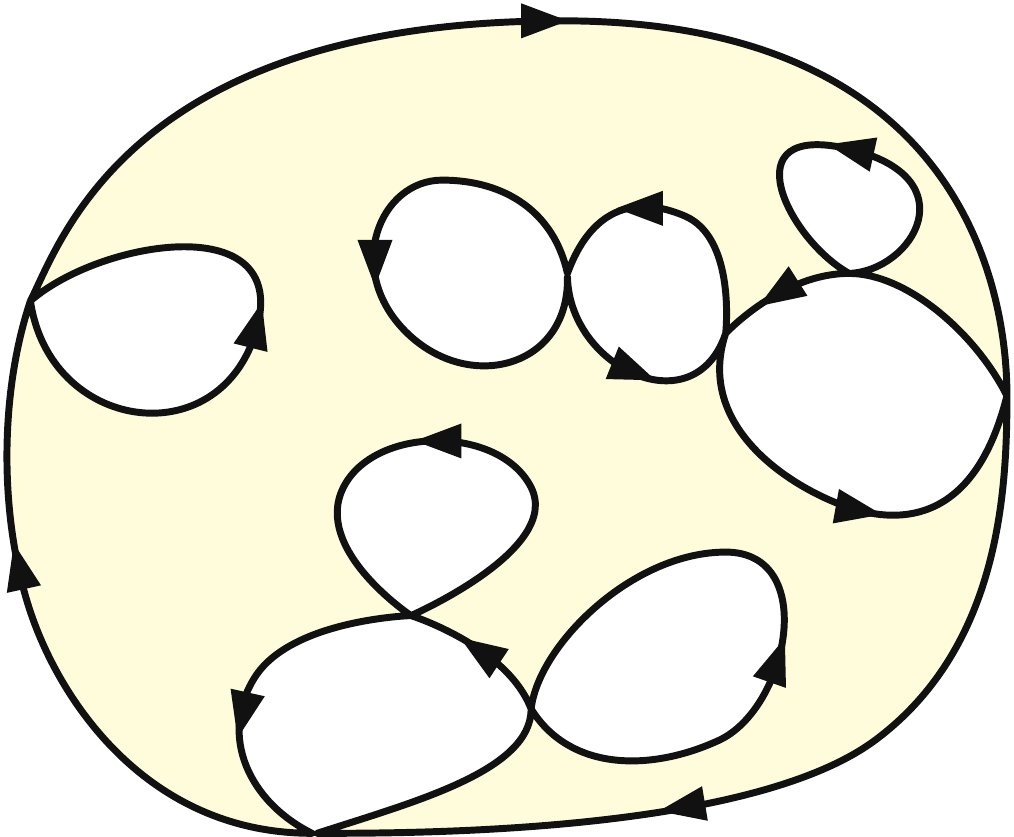}
\caption{A non-simple walk that encloses an open disk.}
\label{F:blob}
\end{figure}

Without loss of generality, assume $\Walk$ is oriented clockwise.  A dual walk that starts at a face enclosed by $\Walk$ cannot visit the same face more than once, because $G$ has no cocycles ($G^\star$ is a dag), and cannot cross $\Walk$, because $\Walk$ is oriented clockwise (all its dual edges point into $A$).  Thus, $\Walk$ must enclose at least one face $f$ with clockwise boundary (a sink in the dual graph $G^\star$).

The interior of~$f$ is a subset of the open disk~$A$, so it must be homeomorphic to an open disk with zero or more holes.  If $f$ is an open disk (and in particular, if $f=A$), we are done.  Otherwise, let $W'$ be closed walk that exactly covers the boundary of any interior hole of $f$.  Because $W'$ lies entirely inside the open disk $A$, it must be contractible.  Thus, $W'$ is a weakly simple contractible closed walk that encloses strictly fewer faces than~$\Walk$, and the lemma follows immediately by induction.
\end{proof}

\begin{theorem}
\label{Th:ccw}
Given a directed graph $G$ embedded on any surface, we can determine in $O(n)$ time whether there is a contractible closed walk in $G$.
\end{theorem}

\begin{proof}
We use essentially the algorithm described by Theorem \ref{Th:con-cycle}: Remove all cocycle edges from $G$ by considering the strong components of the dual graph~$G^\star$, and then examine the faces of the remaining subgraph by brute force.  If we discover a face with coherent boundary whose interior is a disk, we output \textsc{True}; otherwise, we output \textsc{False}.  Correctness follows directly from Lemmas \ref{L:no-cocycle} and~\ref{L:oneface-walk}.
\end{proof}

In the last stage of the proof, if the face $f$ has an interior hole, then the graph $G$ \emph{must} be disconnected, because the interior of $f$ is a subset of the open disk $A$.  Thus, it may tempting to simplify the proof of Lemma \ref{L:oneface-walk} by assuming that the graph $G$ is connected.  Unfortunately, we cannot \emph{simultaneously} assume that $G$ is connected \emph{and} that $G$ contains no cocycles.  Resolving this tension lies at the core of our algorithm for finding bounding closed walks, which we describe next.

\Newpage
\section{Bounding Closed Walks}
\label{S:bcw}
%
%

To simplify our presentation, we assume in this section that our input graphs are embedded on surfaces \emph{without} boundary; this assumption is justified by the following observation.  Let $G$ be a directed graph embedded on a surface $S$ \emph{with} $b>0$ boundary components.  Let $S^\bullet$ be the surface without boundary obtained from $S$ by attaching a single disk with $b-1$ holes to the boundary cycles of $S$.  Thus, every boundary face of $G$ on the original surface $S$ is contained in a single face of $G$ on $S^\bullet$.  If $S$ has first Betti number $\beta$, then $S^\bullet$ has genus $O(\beta)$.

\begin{lemma}
A closed walk $\Walk$ in $G$ is bounding in $S$ if and only if it is bounding in $S^\bullet$.
\end{lemma}

\begin{proof}
Any Alexander numbering for $\Walk$ on $S$ is also an Alexander numbering for $\Walk$ on $S^\bullet$.

Let $\alpha$ be any Alexander numbering for $\Walk$ on $S^\bullet$.  For any integer $k$, the function $\alpha+k$ defined as $(\alpha+k)(f) = \alpha(f) + k$ for every face is also an Alexander numbering for $\Walk$ on $S^\bullet$.  Thus, there is an Alexander numbering $\alpha^\bullet$ of $\Walk$ such that $\alpha^\bullet(f)=0$ for the single face of $G$ containing $S^\bullet\setminus S$. The function~$\alpha^\bullet$ is an Alexander numbering for $\Walk$ on $S$.
\end{proof}

\subsection{Detection}

Our algorithm for detecting bounding closed walks relies on two elementary observations.  First, we can assume without loss of generality that $G$ is strongly connected, because any closed walk in the input graph $G$ stays within a single strong component of $G$.  On the other hand, Lemma \ref{L:no-cocycle} implies that we can assume without loss of generality that the input graph~$G$ has no cocycles; otherwise, we can remove all cocycles in $O(n)$ time by contracting each strong component of the dual graph $G^\star$ to a single dual vertex.  Indeed, if the input graph satisfies \emph{both} of these conditions, detecting bounding closed walks is trivial.

\begin{lemma}
\label{L:bcw-both-done}
Let $G$ be a strongly connected surface graph with at least one edge, whose dual graph $G^\star$ is acyclic.  There is a bounding closed walk in $G$.
\end{lemma}

\begin{proof}
For each face $f$ of $G$, let $\alpha(f)$ be the rank of the dual vertex $f^\star$ in an arbitrary topological sort of the dual graph $G^\star$.  For each edge $e$ in $G$, we have $\alpha(\Left(e)) > \alpha(\Right(e))$.  Thus, the function $\bdry\alpha \colon e \mapsto \alpha(\Left(e)) - \alpha(\Right(e))$ is a positive integer boundary circulation with connected support.  Any Euler tour of this circulation is a bounding closed walk.
\end{proof}

Kao and Shannon \cite{ks-lrgop-89,ks-lpapd2-93,k-lpapd1-93} observed that these two assumptions are actually identical for planar graphs; a planar directed graph $G$ is strongly connected if and only if its dual graph $G^*$ is acyclic.  (This important observation is the basis of several efficient algorithms for planar directed graphs \cite{ks-lrgop-89,ks-lpapd2-93,k-lpapd1-93,kk-totcb-93,ry-fcppg-94,atz-itspd-03}.)  And indeed, finding bounding closed walks in planar graphs is trivial, because \emph{every} closed walk in a planar graph is bounding. 

However, this equivalence does not extend to directed graphs on more complex surfaces; a strongly connected directed surface graph can contain many directed cocycles.  Moreover, deleting all the directed cocycles from a strongly connected surface graph can disconnect the graph.  More subtly, if $H$ is a disconnected surface graph without cocycles, the induced embeddings of the individual components of $H$ \emph{can} contain cocycles.  See Figure~\ref{F:torus-alternate}.

\begin{figure}[ht]
\centering
\includegraphics[scale=0.35]{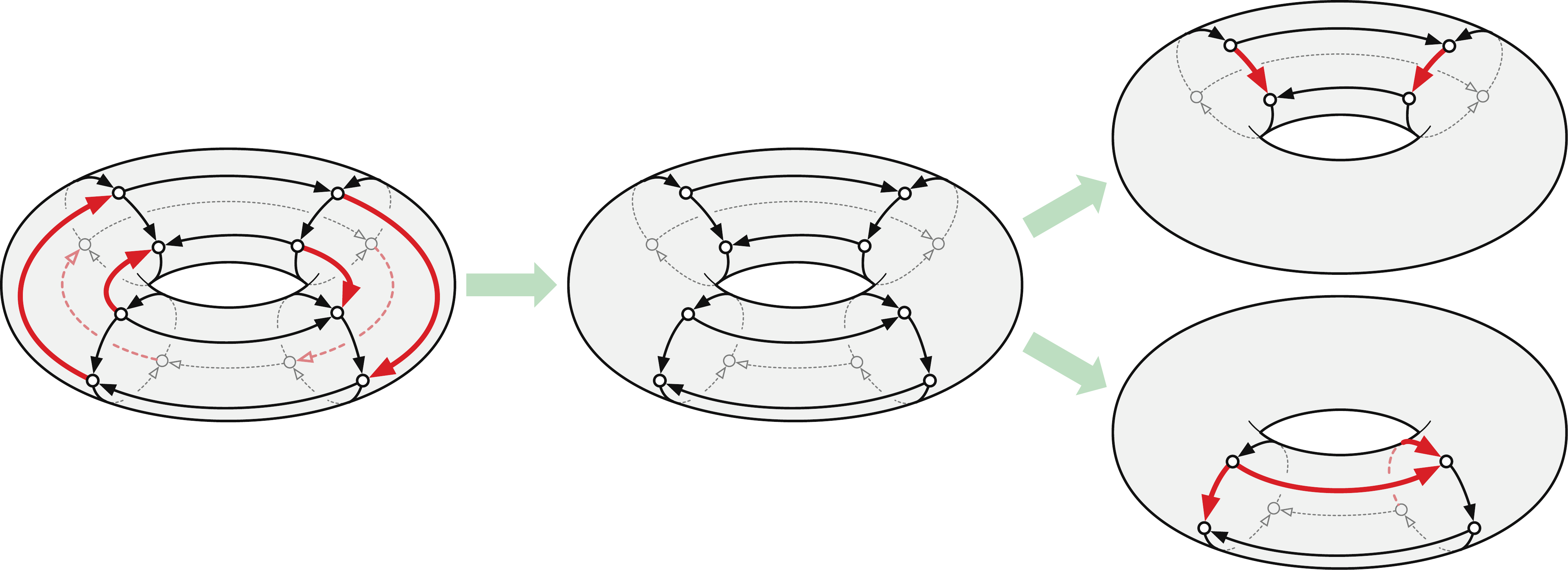}
\caption{A strongly connected directed graph on the torus with cocycles (bold red). Deleting all cocycles disconnects the graph.  The components of the disconnected graph have more cocycles.  (Only one cocycle is emphasized in each component.)}
\label{F:torus-alternate}
\end{figure}

Our algorithm repeatedly applies both of these simplifications, alternately removing cocycle edges and separating components, eventually reporting success if it ever finds a component that is both strongly connected and cocycle-free.  Each iteration can be performed in linear time, and each iteration removes at least one edge from the graph; thus, conservatively, our algorithm runs in $O(n^2)$ time.  But in fact, as we show next, our algorithm converges after only $O(g)$ iterations.

\begin{lemma}
\label{L:strong}
Let $H$ be a directed surface graph without cocycles.  Every weak component of $H$ is strongly connected.
\end{lemma}

\begin{proof}
If a weak component of $H$ is not strongly connected, then $H$ contains a non-empty directed edge cut.  Every non-empty directed edge cut is the disjoint union of cocycles.
\end{proof}

\def\footprint#1{\left\llbracket{#1}\right\rrbracket}

We call any face of $G$ \EMPH{simple} if it is homeomorphic to an open disk and \EMPH{non-simple} otherwise.  (The boundary of a simple face is \emph{not} necessarily a simple cycle.)

\begin{lemma}
\label{L:simple-faces}
Let $G$ be a strongly connected surface graph, let $H$ be the subgraph of $G$ obtained by deleting all cocycles, and let $G'$ be a (strong) component of~$H$.  Every simple face of $G'$ is also a simple face of $G$.
\end{lemma}

\begin{proof}
Let $f$ be an arbitrary simple face of $G'$.  There are two cases to consider.

First, suppose $f$ contains no vertices of $G$ in its interior.  Then any cocycle in $G$ that passes through a face inside $f$ must contain at least one edge on the boundary of $f$.  But by definition, every edge on the boundary of~$f$ is an edge in~$G'$.  It follows that no cocycle in $G$ passes through a face in the interior of $f$, and thus no cocycle in $G$ includes an \emph{edge} in the interior of $f$.  We conclude that $f$ is actually a face of $G$.

On the other hand, suppose some vertex $v$ of $G$ lies inside the open disk $f$.  Then $v$ is disconnected from the boundary of $f$ in $H$ by removing cocycles in~$G$.  None of these cocycles contain an edge on the boundary of~$f$, because all such edges survive in $H$, so~$G$ must contain a cocycle whose edges and faces are entirely inside~$f$.  This cocycle must be a directed edge cut in~$G$, because \emph{every} directed cocycle inside a disk is a directed edge cut.  But no such edge cut exists, because $G$ is strongly connected.
\end{proof}

\begin{corollary}
\label{C:cocycle-nonsimple}
Let $G$ be a strongly connected surface graph, let $H$ be the subgraph of $G$ obtained by deleting all cocycles, and let $G'$ be a (strong) component of~$H$.  Every directed cocycle in $G'$ visits at least one non-simple face of $G'$.
\end{corollary}

We define the \EMPH{footprint $\footprint{G}$} of a directed surface graph $G$ as the union of the vertices, edges, and \emph{simple} faces of~$G$.  For example, the footprint of the middle graph in Figure \ref{F:torus-alternate} is the disjoint union of two annuli.  The embedding of $G$ is cellular if and only if $\footprint{G}$ is the entire surface.  Let \EMPH{$\beta_0(G)$} denote the number of (weak) components of $G$, or equivalently, the number of components of its footprint $\footprint{G}$.  Let \EMPH{$\beta_1(G)$} denote the first Betti number of $\footprint{G}$, which is the rank of the first homology group of $\footprint{G}$.   If~$G$ has at least one non-simple face, then
\[
	\beta_1(G) :=  \beta_0(G) - v(G) + e(G) - f_0(G),
\]
where $v(G)$, $e(G)$, and $f_0(G)$ denote the number of vertices, edges, and simple faces of $G$, respectively.  If~$G$ (and therefore $\footprint{G}$) is disconnected, then $\beta_1(G)$ is the sum of the first Betti numbers of its components.

\begin{lemma}
\label{L:beta-decreases}
Let $G$ be a strongly connected surface graph, let $H$ be the subgraph of $G$ obtained by deleting all cocycles, let $G'$ be a (strong) component of~$H$, and let $H'$ be the subgraph of $G'$ obtained by deleting all cocycles.  If $H' \ne G'$, then $\beta_1(H') < \beta_1(G')$.
\end{lemma}

\begin{proof}
Assume $H'\ne G'$, so $G'$ contains at least one cocycle.  Corollary \ref{C:cocycle-nonsimple} implies that every cocycle in $G'$ passes through a non-simple face of~$G'$.  Thus, every cocycle in $G'$ can be decomposed into edge-disjoint \emph{dual arcs} of the form $f_0\fenceup f_1\fenceup \cdots\fenceup f_k$, where the initial face $f_0$ and final face $f_k$ are non-simple, and the intermediate faces $f_i$ are simple and distinct.

Let $\alpha_1$ be any dual arc in $G'$.  Because $G'$ is strongly connected, $\alpha$ is \emph{not} a directed edge cut, and therefore $G'\setminus \alpha_1$ is (at least weakly) connected.  It follows that $\beta_1(G' \setminus \alpha_1) = \beta_1(G') - 1$.

For any other dual arc $\alpha'$ in $G'$, the edges in $\alpha'\setminus \alpha_1$ define one or more dual arcs in $G'\setminus \alpha_1$.  Thus, we can proceed inductively as follows.  Define a sequence of nested subgraphs $G' = G_0 \supset G_1 \supset G_2 \supset \cdots \supset G_\ell = H$, where each subgraph $G_i$ is obtained by deleting a dual arc $\alpha_i$ from $G_{i-1}$, and the final subgraph $G_\ell = H$ contains no dual arcs (and therefore no cocycles).  For each index $i>0$, either $\alpha_i$ is an edge cut in $G_{i-1}$, which implies $\beta_1(G_i) = \beta_1(G_{i-1})$, or~$\alpha_i$ is not an edge cut, which implies $\beta_1(G_i) = \beta_1(G_{i-1}) - 1$.  The lemma now follows by induction.
\end{proof}

\begin{theorem}
\label{Th:bcw-detect}
Given a directed graph $G$ embedded on any surface with genus $g$, we can determine in $O(gn)$ time whether there is a bounding closed walk in $G$.
\end{theorem}

\begin{proof}
Assume $G$ is strongly connected, since otherwise, we can consider each strong component of~$G$ separately.  Let $H$ be the graph obtained by deleting all cocycles of $G$.  We can construct $H$ in $O(n)$ time by computing the strong components of $G^\star$ in linear time, and then deleting any edge of $G$ whose incident faces lie in the same strong component of $G^\star$.  Lemma \ref{L:no-cocycle} implies that any bounding closed walk in $G$ is also a bounding closed walk in $H$.

If $H$ has no edges, we can immediately report that $G$ has no bounding closed walks.  If $H$ is weakly connected and has at least one edge, then Lemmas \ref{L:bcw-both-done} and \ref{L:strong} imply that $H$ contains a bounding closed walk.  Otherwise, we recursively examine each (strong) component of $H$.

Each vertex of the original input graph $G$ participates in only one subproblem at each level of recursion, so the total time spent at each level of the recursion tree is $O(n)$.  Finally, Lemma \ref{L:beta-decreases} implies that the depth of the recursion tree is at most $O(g)$.
\end{proof}

Call the subgraph obtained from $G$ by alternately removing cocycles and isolating components the \emph{snarl} of $G$; the snarl is the largest subgraph of $G$ in which every component is strongly connected and cocycle-free.  We conjecture that the snarl of a surface graph (or at least one nontrivial component thereof) can actually be extracted in $O(n)$ time.

\subsection{Explicit Construction}
\label{SS:bcw-quad}

This algorithm described in Theorem \ref{Th:bcw-detect} only reports whether $G$ contains a closed bounding walk; it does not actually compute such a walk. With some additional straightforward bookkeeping, we can compute an \emph{implicit} representation of a bounding closed walk as an Alexander numbering $\alpha$ of the faces of~$G$, without increasing the running time of our algorithm.  The complexity of this implicit representation is $O(n)$.  If necessary, we can then obtain an \emph{explicit} bounding closed walk, as an alternating sequence of vertices and edges, by computing an Euler tour of the circulation $\bdry\alpha$, as described in the proof of Lemma~\ref{L:bcw-both-done}.  Because each Alexander number $\alpha(f)$ is an integer between $0$ and $O(n)$, the total length of this Euler tour is $O(n^2)$, and we can compute in $O(n^2)$ time using standard Euler-tour algorithms. 

\begin{corollary}
\label{C:bcw-quadratic-algo}
Let $G$ be a directed surface graph with $n$ vertices.  In $O(n^2)$ time, we can either compute an explicit description of a bounding closed walk in $G$ or report correctly that no such walk exists.
\end{corollary}

\begin{corollary}
\label{C:bcw-quadratic}
Let $G$ be a directed surface graph with $n$ vertices.  The shortest bounding closed walk in~$G$ (if such a walk exists) has length $O(n^2)$.
\end{corollary}

For graphs on the sphere or the torus, every bounding closed walk is actually contractible, so Lemmas~\ref{L:weak-ccw} and \ref{L:weak-once} imply that the shortest bounding closed walk actually has length $O(n)$, and Theorem~\ref{Th:ccw} implies that we can find an explicit bounding closed walk in $O(n)$ time.  For graphs on even slightly more complicated surfaces, however, our quadratic upper bounds are actually tight.

\begin{theorem}
For any fixed integer $g\ge 2$ and any positive integer $n$, there is a directed graph with $2n$ vertices and $3n+1$ edges, embedded on an orientable surface of genus $g$, in which the shortest bounding closed walk has length $\Omega(n^2)$.
\end{theorem}

\begin{proof}
It suffices to consider the case $g=2$.  Let $G$ be a directed graph consisting of a directed cycle of length $n$ (“the long cycle”), a chain of $n$ cycles of length $2$ (“the short cycles”), and a single cycle of length~$1$ (“the self-loop”), embedded on a surface of genus $2$ as shown in Figure \ref{F:2-torus-long-ccw} (for $n=5$).  This embedding has $n+1$ faces, of which two are annuli and the rest are disks.  The long cycle and the self-loop both lie on the common boundary of both annular faces.  

\begin{figure}[ht]
\centering
\includegraphics[scale=0.5]{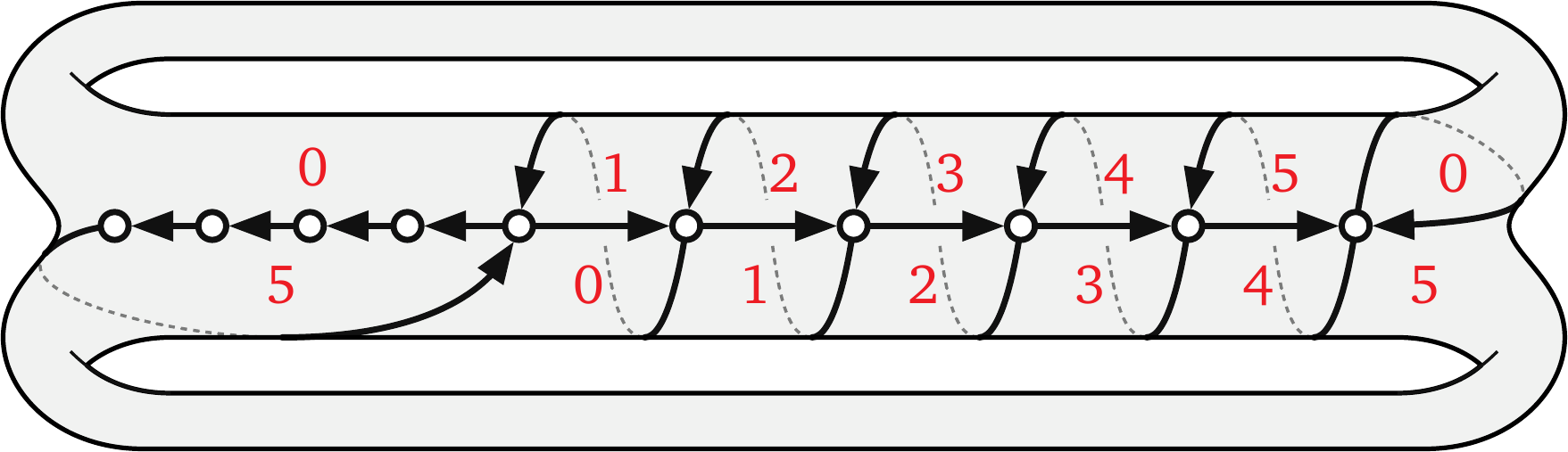}
\caption{A directed surface graph with only long bounding closed walks, with a minimal valid Alexander numbering.}
\label{F:2-torus-long-ccw}
\end{figure}

The dual graph $G^\star$ consists of a directed path $f_0 \arcto f_1 \arcto \cdots \arcto f_{n-1} \arcto f_n$, where $f_0$ and $f_n$ are the annular faces, plus $n$ additional edges directly from $f_0$ to $f_n$.  Because $G^\star$ is a dag, $G$ has no cocycles, and so Lemma~\ref{L:bcw-both-done} implies that there is a bounding closed walk in $G$.  Specifically, the function $\alpha(f_i) = i$ for all $i$ is the Alexander numbering of a bounding circulation; again, see Figure~\ref{F:2-torus-long-ccw}.

The subgraph of edges traversed by any closed walk is strongly connected.  Straightforward exhaustive case analysis implies that the only strongly connected subgraph of $G$ that contains both the long cycle and the self-loop is the entire graph $G$.  Thus, every strongly connected proper subgraph $H$ of $G$ omits either the entire long cycle or the self-loop.  In either case, faces $f_0$ and $f_n$ of $G$ belong to a single face of $H$.  Thus, the dual graph~$H^\star$ is strongly connected, and Lemma \ref{L:no-cocycle} implies that $H$ cannot support a bounding closed walk.  We conclude that every bounding closed walk in $G$ traverses \emph{every} edge of~$G$ at least once.

Now let $\Walk$ be any bounding closed walk in $G$ that traverses every edge at least once, and let $\alpha \colon F(G) \to \Z$ be the Alexander numbering of $\Walk$.  For each edge $e$, we must have $\alpha(\Left(e)) > \alpha(\Right(e))$, so $\alpha$ must be consistent with the only topological order of $G^\star$.  It follows that $\alpha(f_n) - \alpha(f_0) \ge n$.  Thus, $\Walk$ must traverse both the long  cycle and the self-loop $n$ times, and every other edge of $G$ at least once.  We conclude that $\Walk$ must have length at least $n^2 + 3n$.  (This is exactly the length of any  bounding closed walk induced by the Alexander numbering $\alpha(f_i)=i$, so our analysis is tight.)
\end{proof}

One frustrating source of complexity in our construction is the long cycle of degree-2 vertices.  A~different construction with no degree-$2$ vertices yields a lower bound of $\Omega(gn)$ on any surface of genus $g$; we leave this construction as an amusing exercise for the reader, and a proof or disproof of its optimality as an open problem.

\Newpage
\section{NP-hardness}
\label{S:NP-hard}

%

\subsection{Simple Bounding Cycles and Shortest Bounding Walks}
\label{AS:bounding-hard}

Our algorithms for detecting contractible closed walks and simple contractible cycles are essentially identical: Remove all cocycles, and then check the boundary of every face.  Unfortunately, this similarity does not extend to bounding walks and bounding cycles.  In this section, we prove that determining whether a directed surface graph contains a simple bounding cycle is NP-hard.  The same reduction also implies that finding the shortest bounding closed walk in a directed surface graph is NP-hard, even if all edges have weight $1$.

Our proof uses a variant of the “heaven and hell” NP-hardness reduction of Cabello \etal~\cite{ccl-fctpe-11} for the corresponding problem in undirected graphs, which is based in turn on a reduction of Chambers \etal~\cite{splitting} for the problem of finding the shortest closed walk in an undirected surface graph that is separating but non-contractible.  In fact, our reduction is simpler than either of these earlier arguments, because our underlying graphs are directed.

We reduce from the Hamiltonian cycle problem in directed planar graphs, where every vertex has either in-degree $1$ and out-degree $2$, or in-degree $2$ and out-degree $1$.  This problem was proved NP-hard by Plesńik~\cite{p-nhcpp-79}, using a direct reduction from \textsc{3Sat}.

Let $G$ be a planar directed graph.  Fix a planar embedding of $G$ on the sphere (“\emph{Earth}”).  Ultimately, we will construct an embedding of \emph{the same graph $G$} on a more complex surface~$S$, such that $G$ has a Hamiltonian cycle if and only if $G$ has a closed walk of length at most $n$ that is bounding in $S$.

As a preliminary step, we partially color the faces of $G$ as follows.  Call an edge of $G$ \EMPH{forced} if it is the only edge leaving a vertex or the only edge entering a vertex (or both).  Color the left shore of every forced edge blue, and color the right shore of every forced edge red.  A face may be colored both red and blue, or not colored at all; see Figure \ref{F:planar-ham-coloring}.

\begin{figure}[ht]
\centering
\includegraphics[scale=0.4]{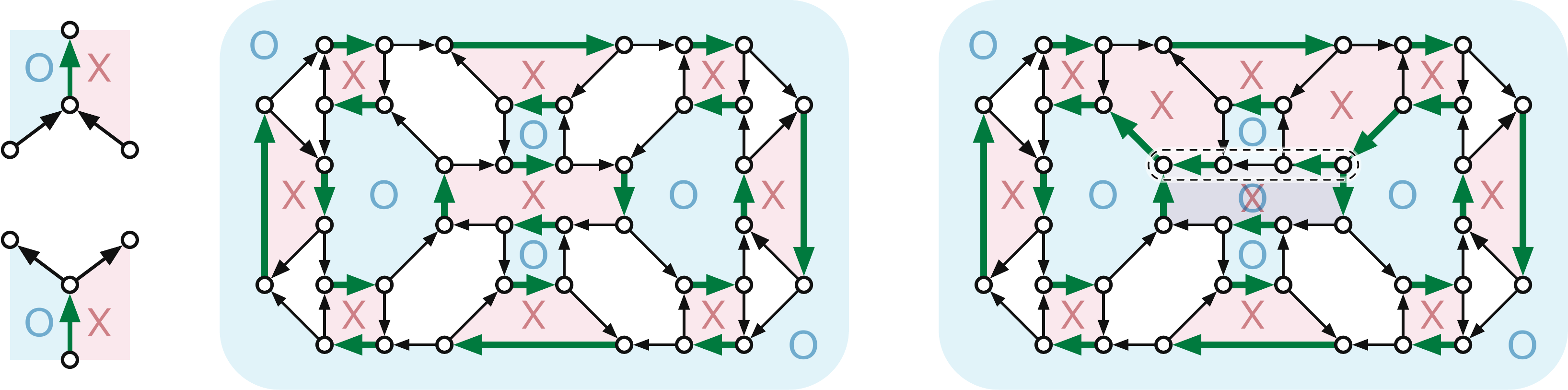}
\caption{Coloring the faces of cubic directed planar graphs; bold green edges are forced.  The graphs differ only in the direction of three edges at the top of the central face.  The second graph is not Hamiltonian, because its central face is colored both red and blue.}
\label{F:planar-ham-coloring}
\end{figure}

\begin{lemma}
If $G$ has a Hamiltonian cycle, then no face of $G$ is colored both red and blue.
\end{lemma}

\begin{proof}
Suppose $G$ contains a Hamiltonian cycle $\gamma$.   By definition, $\gamma$ visits each vertex of $G$ exactly once, so $\gamma$ is a simple closed curve that traverses every forced edge in $G$.  Thus, each blue face is on the left side of $\gamma$, and each red face is on the right side of $\gamma$. The lemma now follows immediately from the Jordan curve theorem.
\end{proof}

Assuming no face of $G$ is colored both red and blue, we construct the surface $S$ in two stages.  First we delete a small disk from the interior of each blue face and attach the resulting punctures to another punctured sphere, thereby merging all blue faces into a single blue face, which we call \emph{heaven}.  Similarly, we delete a small disk from the interior of each red face and attach the resulting holes to a single punctured sphere, merging all the red faces into a single red face, which we call \emph{hell}.  Euler's formula implies that the resulting surface $S$ has genus $O(n)$.
%

\begin{lemma}
The following statements are equivalent:  (a)~$G$ contains a Hamiltonian cycle.  (b)~$G$~contains a simple cycle that is bounding in $S$.  (c)~There is a closed walk in $G$ that is bounding in $S$ and has length at most $n$.
\end{lemma}

\begin{proof}
We prove the implication (a)$\Rightarrow$(b) first.  Suppose $\gamma$ is a Hamiltonian cycle in $G$.  Viewed as a cycle in the original planar embedding, $\gamma$ separates all blue face from all red faces.  Thus, the embedding of $\gamma$ on $S$ separates heaven from hell; it follows that $\gamma$ is a bounding cycle in $S$.

The implication (b)$\Rightarrow$(c) is trivial.

Finally, we prove the implication (c)$\Rightarrow$(a).  Suppose $\Walk$ is a closed walk in $G$ that is bounding in~$S$ and has length at most $n$.  Fix an Alexander numbering $\alpha$ of the faces of $G$ on $S$ that is consistent with~$\Walk$.  Because $\Walk$ visits at least one vertex of $G$, it traverses at least one forced edge.  It follows that $\alpha(\emph{heaven}) \ne \alpha(\emph{hell})$, which implies that $\Walk$ traverses \emph{every} forced edge, and thus visits every vertex.  We conclude that $\Walk$ is a Hamiltonian cycle.
\end{proof}

\begin{theorem}
Deciding whether a directed surface graph contains a simple bounding cycle is NP-hard.
\end{theorem}

\begin{theorem}
Computing the shortest bounding closed walk in a directed surface graph is NP-hard.
\end{theorem}

Both of these NP-hardness results hold even if we insist on cellularly embedded graphs.  Without loss of generality, suppose the original input graph $G$ is connected, since otherwise it cannot contain a Hamiltonian cycle.  It follows that heaven and hell are the only non-disk faces in $G$'s embedding on~$S$.  Thus, we can extend $G$ to a cellularly embedded graph $H$ by adding a single vertex $a$ in the interior of heaven, with directed edges $\arc{v}{a}$ from one vertex $v$ on the boundary of each blue face, and a single vertex $z$ in the interior of hell, with edges $\arc{v}{z}$ from one vertex $v$ on the boundary of each red face.  Because the two new vertices are sinks, they cannot appear in any (bounding) closed walk in~$H$.

\subsection{Shortest Trivial Cycles}
\label{AS:CCW-hard}

Now we prove that finding the shortest simple contractible cycle in a directed surface graph is NP-hard, by reduction from the classical maximum independent set problem.  Our reduction is nearly identical to a reduction of Cabello~\cite{c-fscss-10}, who proved that finding the shortest simple \emph{bounding} cycle in an \emph{undirected} surface graph is NP-hard.  In the same paper, Cabello described an algorithm to compute the shortest contractible cycle in an \emph{undirected} surface graph in $O(n^2\log n)$ time; thus, our NP-hardness proof must make essential use of the fact that the input graph is directed.

Given an undirected graph $G$ with $n$ vertices and $m$ edges and a positive integer $k$, we first construct a directed planar graph $H$ as follows.  Arbitrarily identify the $n$ vertices of $G$ with the integers $0$ through $n-1$.  For each integer $1\le i\le n$ the graph $H$ contains a corresponding vertex $i$, along with two directed paths $P_i$ and $Q_i$ from vertex $i$ to vertex ${i+1 \bmod n}$.  Path $P_i$ has length $\deg(i)+1$, and path $Q_i$ has length $\deg(i)+2$.  Label the vertices of $P_i$ as $i \arcto [ij_1] \arcto [ij_2] \arcto\cdots \arcto [ij_{\deg(i)}] \arcto (i+1) \bmod n$, where $j_1, j_2, \dots, j_{\deg(i)}$ are the neighbors of vertex $i$ in $G$ (in no particular order).  Thus, for each edge $ij$ of $G$, there are two corresponding vertices in $H$, namely $[ij]$ in $P_i$ and $[ji]$ in $P_j$.  In total, $H$ has $2n+4m$ vertices and $3n+4m$ edges.

We embed $H$ on the sphere so that one face is bounded by the paths $P_0, P_1, \dots, P_{n-1}$, another face is bounded by the paths $Q_0, Q_1, \dots, Q_{n-1}$, and each of the remaining faces is bounded by $P_i$ and $Q_i$ for some index $i$.  See Figure \ref{F:short-con}.

\begin{figure}[ht]
\centering\includegraphics[scale=0.5]{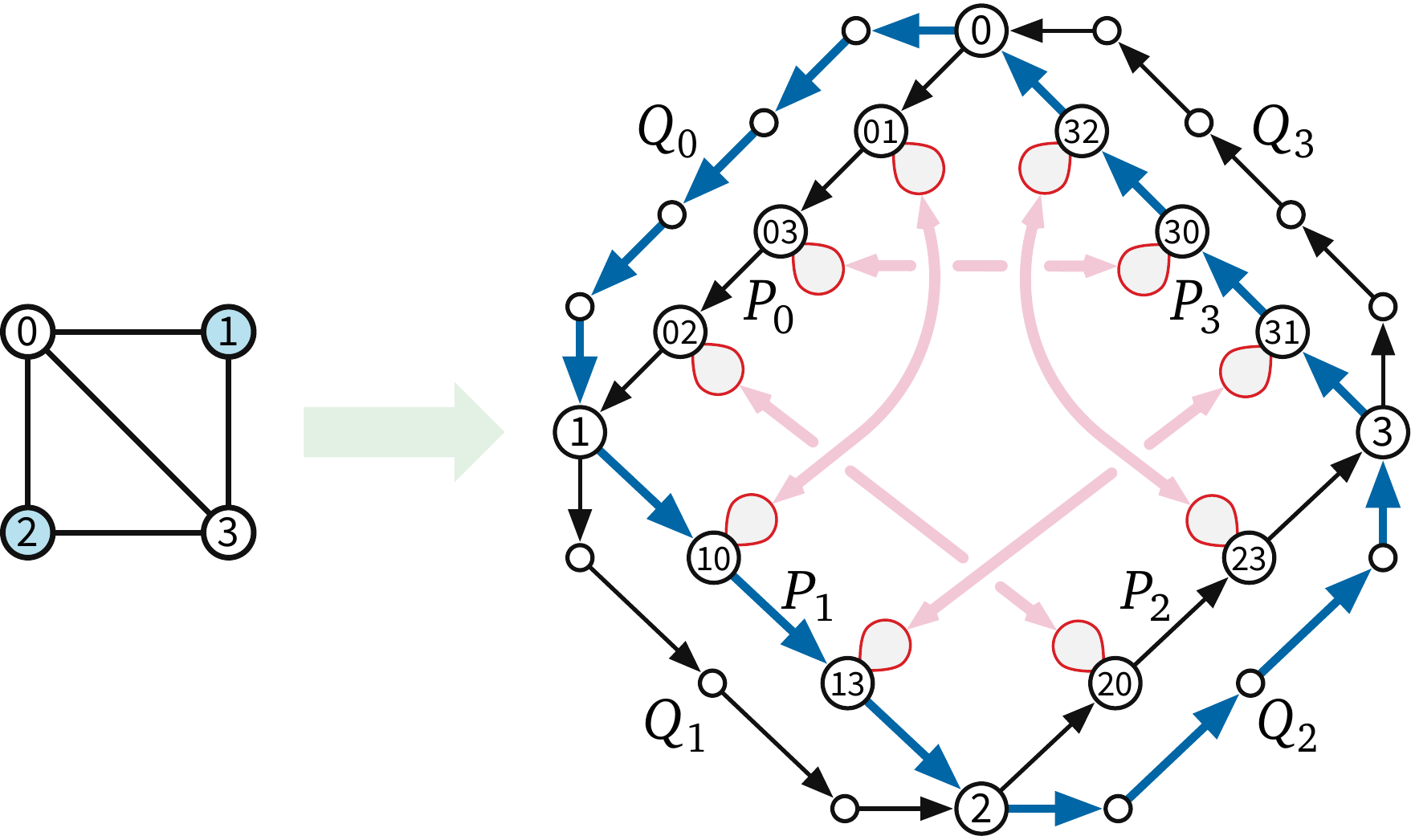}
\caption{Reduction from maximum independent set to shortest contractible cycle in a directed surface graph.}
\label{F:short-con}
\end{figure}

A simple cycle $\gamma$ in $H$ is {\em compliant} if, for every edge $ij$ in $G$, $\gamma$ does not use both of the vertices $[ij]$ and $[ji]$ in $H$.

\begin{lemma}
$G$ contains an independent set of size $k$ if and only if $H$ contains a simple compliant cycle of length at most $2m + n - k$.
\end{lemma}

\begin{proof}
First, suppose $G$ has an independent set $I$ of size $k$.  For each integer $0\le i\le n-1$, let $R_i = P_i$ if $i \in I$, and define $R_i = Q_i$ otherwise.  Let $\gamma$ be the simple cycle formed by concatenating the paths $R_0, R_1, \dots, R_{n-1}$.  This cycle is compliant, because there are no edges in $G$ between the vertices in $I$.  Straightforward calculation implies that the length of $\gamma$ is
$2m + n - k$.

On the other hand, suppose $H$ contains a compliant cycle $\gamma$ of length $2m + n - k$.  Every cycle in $H$ is the concatenation of paths $R_0, R_1, \dots, R_{n-1}$, where for each index $i$, either $R_i=P_i$ or $R_i=Q_i$.  Because $\gamma$ is compliant, $\gamma$ cannot contain both $P_i$ and $P_j$, for any edge $ij$ of $G$.  Thus, the set of indices $I = \set{i \mid R_i = P_i}$ is an independent set in $G$.   Straightforward calculation implies $\abs{I} = k$.
\end{proof}
    
We now convert $H$ into a directed graph $H'$ embedded on a more complex surface $S$ as follows.  Let~$f$ be the face of $H$ bounded by paths $P_0, P_1, \dots, P_{n-1}$.  For each edge $ij$ of the original graph $G$, let $\alpha_{ij}$ and $\alpha_{ji}$ denote two small simple closed curves, respectively passing through vertices $[ij]$ and $[ji]$ of $H$, and otherwise in the interior of $f$.  To define both the new graph $H'$ and the surface $S$, we remove the interiors of all curves $\alpha_{ij}$ and then identify each curve $\alpha_{ij}$ with its twin $\alpha_{ji}$, so that  vertices $[ij]$ and $[ji]$ are identified and the resulting surface $S$ is orientable.  The resulting surface $S$ has genus $m$.

\begin{lemma}
\label{L:compliant}
For any integer $\ell$, there is a simple compliant cycle of length $\ell$ in $H$ if and only if there is a simple cycle of length $\ell$ in $H'$ that is contractible in $S$.
\end{lemma}

\begin{proof}
First, let $\gamma$ be \emph{any} simple cycle in $H$.  The edges in $H'$ are in one-to-one correspondence with edges in $H$, and sequence of edges corresponding to $\gamma$ defines a closed walk $\gamma'$ in $H'$ with the same length at $\gamma$.  Because the construction of $H'$ and $S$ only involves surgery on the interior of the single face $f$, the corresponding closed walk $\gamma'$ in $H$ is contractible in $S$.  Finally, if $\gamma$ is compliant, then $\gamma'$ is simple.

Now let $\gamma'$ be a simple contractible cycle in $H'$.  For each edge $ij$ of the original graph $G$, the curve $\alpha_{ij} = \alpha_{ji}$ intersects $H'$ at exactly one vertex $[ij] = [ji]$ and therefore crosses $\gamma'$ at most once.  Because~$\gamma'$ is contractible, $\gamma'$ crosses \emph{every} closed curve in $S$ an even number of times.  It follows that $\gamma'$ does hot cross any curve $\alpha_{ij}$.  Thus, the corresponding sequence of edges in $H$ defines a simple cycle $\gamma$.  Finally, because~$\gamma'$ visits each vertex $[ij]=[ji]$ at most once, $\gamma$ is compliant.
\end{proof}

\begin{theorem}
Finding the shortest simple contractible cycle in a directed surface graph is NP-hard.
\end{theorem}

The statement and proof of Lemma \ref{L:compliant} still hold without further modification if we replace every instance of the word “contractible” with the word “bounding” or “separating”.  Indeed, every bounding cycle in $H'$ is actually contractible. 

\begin{theorem}
\label{Th:shortest-sbc-hard}
Finding the shortest simple bounding cycle in a directed surface graph is NP-hard.
\end{theorem}

\Newpage
\section{Shortest Contractible Closed Walks: The Easy Cases}
\label{S:short-con-easy}

\noindent
Now we describe polynomial-time algorithms to find a shortest contractible closed walk in a directed surface graph with weighted edges.  For this problem, it will prove convenient to treat input graphs as undirected, but with asymmetric edge weights, as proposed by Cabello~\etal~\cite{ccl-fsncd-16}.  That is, each undirected edge $uv$ in the input graph is composed of two \emph{darts} $\arc{u}{v}$ and $\arc{v}{u}$ with independent weights; a walk is alternating series of vertices and \emph{darts}; and the length of a walk is the sum of the weights of its darts, counted with appropriate multiplicity.   We also assume without loss of generality that the given embedding of $G$ is \emph{cellular}, meaning every face is homeomorphic to an open disk.  Both of these assumptions can be enforced by adding at most $O(n+\beta)$ directed edges with (symbolically) infinite weight, where $\beta$ is the first Betti number of the underlying surface.  Finally, to avoid issues with negative cycles, we assume all edge weights are non-negative. 

All orientable 2-manifolds can be partitioned into four classes, depending on the structure of their fundamental group.  
\begin{itemize}\itemsep0pt
\item
The sphere and the disk have \emph{trivial} fundamental groups.
\item
The annulus and the torus have \emph{abelian} fundamental groups.
\item
Every surface with boundary has a \emph{free} fundamental group.
\item
Every surface without boundary and with genus at least $2$ has a \emph{hyperbolic} fundamental group.
\end{itemize}
All closed curves on the sphere or the disk are contractible, so our problem is trivial for those surfaces.  We develop a separate polynomial-time algorithm for each of the other three classes.

Our algorithm for the annulus and the torus, described in Section \ref{SS:ccw-abelian}, constructs a finite relevant portion~$\Relevant{G}$ of the universal cover of the input graph, and then finds the shortest directed cycle in~$\Relevant{G}$ using a recent algorithm of Mozes \etal\ \cite{mnnw-mdpgo-18}.  Our algorithm for surfaces with boundary, described in Section~\ref{SS:ccw-free}, exploits an observation by Smikun \cite{s-cbcgg-76} and Muller and Schupp \cite{ms-gtecl-83} that the set of trivial words for any finitely-generated free group is a context-free language.  After carefully labeling the input graph, we invoke an algorithm of Barrett \etal~\cite{bjm-flcpp-00} that finds the shortest walk in a given graph whose structure is consistent with a given context-free grammar.  We defer discussion of the hyperbolic case to Section~\ref{S:short-con-hyperbolic}.

\subsection{Annulus and Torus (Abelian)}
\label{SS:ccw-abelian}

The universal cover $\Lift{G}$ of a surface graph $G$ is an infinite planar graph, constructed by tiling the plane with an infinite number of copies of a disk called the \EMPH{fundamental domain}, obtained by cutting the underlying surface along a system of loops \cite{e-dgteg-03,k-csnco-06} or a system of arcs \cite{en-mcsnc-11}.  A closed walk in $G$ is contractible if and only if it is the projection of a closed walk in $\Lift{G}$.  Thus, the shortest contractible walk in~$G$ is the projection of the shortest directed cycle in $\Lift{G}$.

Corollary \ref{C:ccw-short} implies that we can find the shortest contractible closed walk in $G$ by searching a relevant region $\Relevant{G}$ of $\Lift{G}$, assembled from a finite copies of the fundamental domain, and just large enough to contain at least one lift of \emph{every} closed walk of length $m = O(n)$ in $G$.  The assembled graph~$\Relevant{G}$ is planar, so we can find the shortest directed cycle in $O(\Relevant{n}\log\log\Relevant{n})$ time, where $\Relevant{n}$ is the number of vertices in $\Relevant{G}$, using the recent algorithm of Mozes \etal\ \cite{mnnw-mdpgo-18}. 

This approach immediately gives us efficient algorithms when the fundamental group of the underlying surface is \emph{abelian}.  There are only two such surfaces:

\begin{theorem}
Given a directed graph $G$ with $n$ vertices and non-negatively weighted edges embedded on the annulus, we can find the shortest contractible closed walk in $G$ in $O(n^2\log\log n)$ time.
\end{theorem}

\begin{proof}
We can compute an appropriate fundamental domain in $O(n)$ time by slicing the annulus along any path in $G$ between the two boundaries of the annulus.  $\Relevant{G}$ consists of at most $O(n)$ copies of the fundamental domain and thus has complexity $O(n^2)$.
\end{proof}

\begin{theorem}
Given a directed graph $G$ with $n$ vertices and non-negatively weighted edges embedded on the torus, we can find the shortest contractible closed walk in $G$ in $O(n^3\log\log n)$ time.
\end{theorem}

\begin{proof}
We can compute an appropriate fundamental domain in $O(n)$ time by slicing the surface along any system of loops \cite{e-dgteg-03,k-csnco-06}.  $\Relevant{G}$ consists of at most $O(n^2)$ copies of the fundamental domain and thus has complexity $O(n^3)$.
\end{proof}

In principle, we can use the same covering-space approach to find shortest contractible closed walks in graphs on \emph{arbitrary} surfaces.  Unfortunately, for every nontrivial surface other than the annulus and the torus, covering all paths of length $m$ requires exponentially many copies of the fundamental domain.%

\subsection{Surfaces with Boundary (Free)}
\label{SS:ccw-free}

The fundamental group of a surface with genus $g$ and $b>0$ boundary cycles is a \emph{free} group with $\beta = 2g + b - 1$ generators.  Smikun \cite{s-cbcgg-76} and Muller and Schupp \cite{ms-gtecl-83} observed that the set of trivial words for any finitely-generated free group is a context-free language.%
\footnote{More generally, both Smikun and Muller and Schupp proved that the set of trivial words for a finitely-presented group is a context-free language if and only if the group has a free subgroup with finite index.} 

This observation allows us to reduce finding shortest contractible walks on surfaces with boundary to the \EMPH{CFG shortest path} problem, introduced by Yannakakis~\cite{y-gmdt-90}.  Given a directed graph $G$ with weighted edges, a labeling $\ell\colon E(G)\to \Sigma\cup\set{\e}$ for some finite set $\Sigma$, and a context-free grammar~$\Grammar$ over the alphabet~$\Sigma$, the CFG shortest path problem asks for the shortest walk (if any) in~$G$ whose label is in the language generated by $\Grammar$.  An algorithm of Barrett \etal~\cite{bjm-flcpp-00} solves the CFG shortest path problem in $O(N P n^3)$ time, where $N$ and $P$ are the numbers of nonterminals and productions in~$\Grammar$, respectively, and $n$ is the number of vertices in $G$, when the grammar~$\Grammar$ is in Chomsky normal form and all edge weights in~$G$ are non-negative.  (Their algorithm was extended to graphs with negative edges but no negative cycles by Bradford and Thomas~\cite{bt-lspdw-09}.)

\begin{theorem}
Given a directed graph $G$ with $m$ non-negatively weighted edges, embedded on a surface~$S$ with boundary with first Betti number $\beta$, we can find the shortest contractible closed walk in $G$ in $O(\beta^5 m^3)$ time.
\end{theorem}

\begin{proof}
Our algorithm begins by constructing a \emph{system of dual arcs} for $G$—a collection of boundary-to-boundary paths $\alpha_1, \alpha_2, \dots, \alpha_\beta$ in $G^*$ that cut the underlying surface $S$ into a disk.  Such a system can be constructed in $O(gn)$ time \cite{en-mcsnc-11}, using a natural variant of Eppstein's tree-cotree construction of systems of loops \cite{e-dgteg-03}.  If necessary, we subdivide the edges of $G$ (introducing parallel edges into $G^*$) so that these dual arcs are simple and edge-disjoint, as described by Cabello and Mohar \cite{cm-fsnsn-07} for systems of loops.  After subdivision, $G$ has at most $\beta m$ edges and therefore (because $G$ is symmetric) at most $\beta m$ vertices.  We arbitrarily direct each arc $\alpha_i$.  Finally, we label each directed edge of $G$ with a generator $\Sym{a}_i$ if it crosses the corresponding dual arc $\alpha_i$ from left to right, with an inverse generator $\Sym{\Inverse{a}}_i$ if it crosses~$\alpha_i$ from right to left, and with the empty string otherwise.

Any closed walk in $G$ is naturally labeled with a string obtained by concatenating the labels of its edges in order.  A closed walk is contractible if and only if its label can be \emph{reduced} to the empty string by repeatedly removing substrings of the form $\Sym{a\Inverse{a}}$ or \Sym{\Inverse{a}a}, where \Sym{a} is a generator and \Sym{\Inverse{a}} is its inverse.  The set of all such reducible strings is generated by the following Chomsky normal form context-free grammar $\Grammar$:
\begin{align*}
	\Sym{I} &\to \Sym{II}
			\mid \Sym{A}_1\Sym{\Inverse{A}}_1 \mid \Sym{\Inverse{A}}_1 \Sym{A}_1
			\mid \Sym{A}_2\Sym{\Inverse{A}}_2 \mid \Sym{\Inverse{A}}_2 \Sym{A}_2
			\mid \cdots
			\mid \Sym{A}_\beta\Sym{\Inverse{A}}_\beta \mid
						\Sym{\Inverse{A}}_\beta \Sym{A}_\beta
\\
	\Sym{A}_i &\to \Sym{a}_i \mid \Sym{A}_i\Sym{I} \mid \Sym{I}\Sym{A}_i
		\qquad\qquad\qquad\qquad\qquad \text{for all $i$}
\\
	\Sym{\Inverse{A}}_i &\to \Sym{a}_i \mid \Sym{\Inverse{A}}_i\Sym{I} \mid \Sym{I}\Sym{\Inverse{A}}_i
		\qquad\qquad\qquad\qquad\qquad \text{for all $i$}
\end{align*}
Specifically, the starting nonterminal \Sym{I} generates reducible strings; each non-terminal $\Sym{A}_i$ generates strings that reduce to the corresponding generator $\Sym{a}_i$;  each barred non-terminal $\Sym{\Inverse{A}}_i$ generates strings that reduce to the corresponding inverse generator $\Sym{\Inverse{a}}_i$.

We are now left with instance of the CFG shortest path problem with $N = 2\beta+1$ non-terminals and $P = 8\beta+1$ productions in a graph with $n = O(\beta m)$ vertices, which can be solved in $O(N P n^3) = O(\beta^5 m^3)$ time by the algorithm of Barrett \etal~\cite{bjm-flcpp-00}.
\end{proof}


\section{Shortest Contractible Closed Walks: The Hyperbolic Case}
\label{S:short-con-hyperbolic}

Finally, we describe our algorithm to compute shortest contractible closed walks on surfaces with no boundary and with genus at least $2$.  All such surfaces have \emph{hyperbolic} fundamental groups in the sense of Gromov \cite{g-hg-87}; as observed by Dehn over a century ago \cite{d-uudg-11}, the most natural geometry for the universal cover of these surfaces is the hyperbolic plane.  (See Figure \ref{F:universal-cover} below.)  

Neither of the algorithms described in Section \ref{S:short-con-easy} apply to graphs on hyperbolic surfaces.  As we already mentioned, adapting the covering-space algorithm from Section \ref{SS:ccw-abelian} would require exponentially many copies of the fundamental domain.  (See Lemma \ref{L:exp-growth} below.)  On the other hand, the set of trivial words for a hyperbolic group is \emph{not} a context-free language~\cite{ms-gtecl-83}, so we cannot apply our algorithm from Section~\ref{SS:ccw-free} directly.  Nevertheless, we solve our problem by reduction to the CFG-shortest path algorithm of Barrett \etal~\cite{bjm-flcpp-00}, by constructing a grammar that generates all \emph{sufficiently short} contractible closed walks in a certain canonical surface map. 

Fix an undirected graph $G$ embedded on an orientable surface $S$ with genus $g\ge 2$ and no boundary.  As in the previous section, we assume without loss of generality that $G$ is symmetric, the embedding of $G$ is cellular, and all edge weights are non-negative.

\subsection{Edge Labeling}

The first step in our algorithm is to label the edges of $G$ so that the labels along any closed walk encode its homotopy type.  In principle, we can derive such a labeling by reducing $G$ to a system of loops \cite{dg-tcs-99,cm-fsnsn-07,k-csnco-06}; however, this labeling leads to a less efficient algorithm.\footnote{See footnote \ref{fn:whyquads}.}  Instead, we reduce $G$ to a different canonical surface decomposition called a \emph{system of quads}, as proposed by Lazarus and Rivaud \cite{lr-hts-12} and Erickson and Whittlesey \cite{ew-tcsr-13}.

Let $(T, L, C)$ be an arbitrary \emph{tree-cotree decomposition} of $G$, which partitions the edges of $G$ into a spanning tree $T$, a spanning tree $C^*$ of the dual graph $G^*$, and exactly $2g$ leftover edges $L$ \cite{e-dgteg-03}.  Contracting every edge in $T$ and deleting every edge in $C$ reduces $G$ to a \emph{system of loops}, which has one vertex $a$, one face $f$, and $2g$ loops $L$.  To construct the system of quads $Q$, we introduce a new vertex $z$ in the interior of $f$, add edges between $z$ and every corner of $f$, and then deleting the edges in $L$.

Next we label each directed edge $e$ in $G$ with a directed walk \EMPH{$\seq{e}$} in $Q$ as follows.  Every walk $\seq{e}$ starts and ends at $a$ and thus has even length.  If $e\in T$, then $\seq{e}$ is the empty walk.  Otherwise, after contracting~$T$, edge $e$ connects two corners of~$f$; we define $\seq{e}$ as the walk of length $2$ in $Q$ from the back corner of $e$, to $z$, and then to the front corner of $e$.  (If $e\in L$, there are two possibilities for $\seq{e}$; choose one arbitrarily.)

\begin{figure}[ht]
\centering
	\includegraphics[scale=0.4]{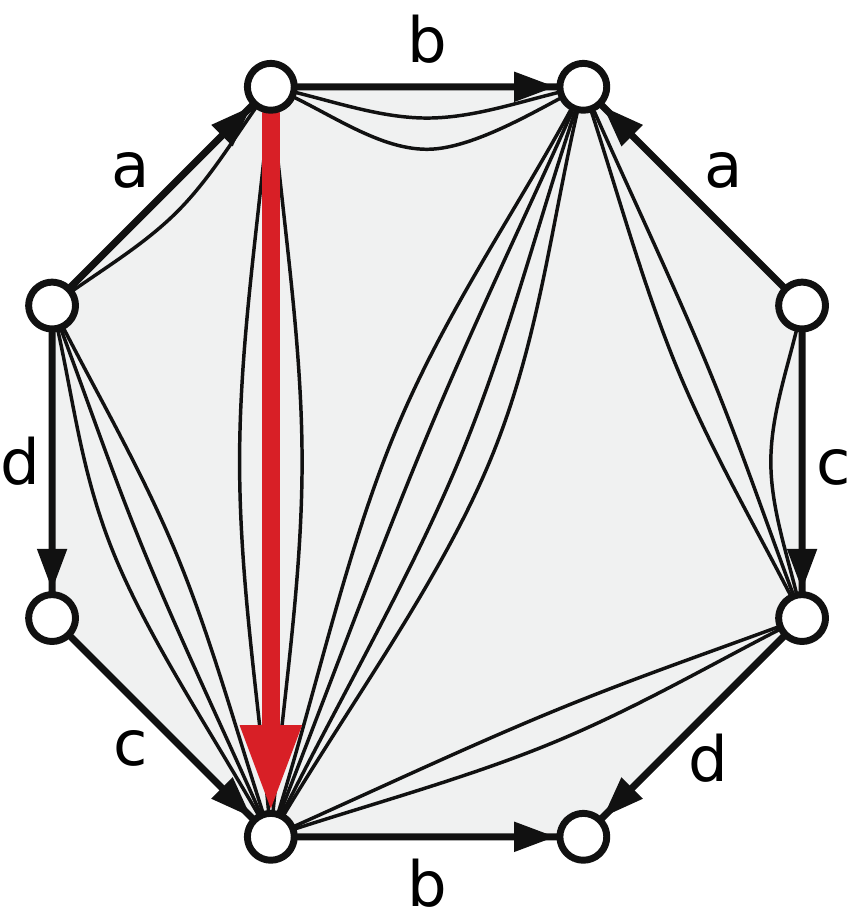} \qquad\qquad
	\includegraphics[scale=0.4]{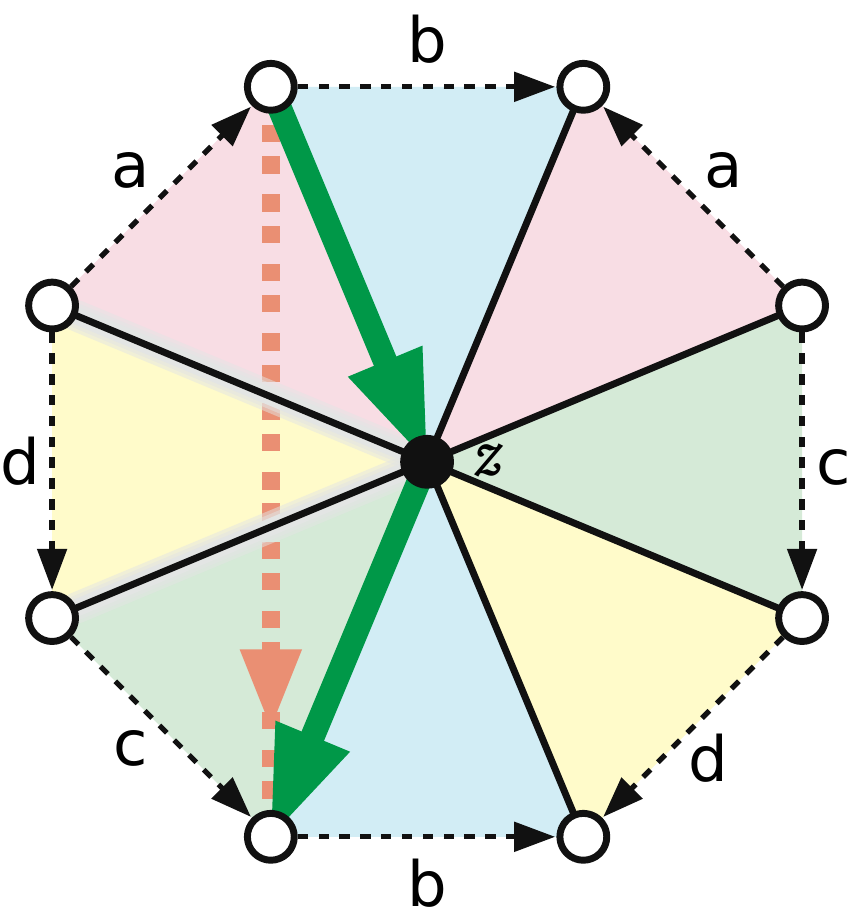}
\caption{Left: A fundamental polygon for a surface of genus~$2$, obtained by contracting a spanning tree and cutting along a system of loops.  Right: Mapping a directed non-tree edge to a path of length $2$ in the corresponding system of quads.}
\label{F:reduction}
\end{figure}

Finally, for any closed walk $\Walk$ in $G$, let $\seq{\Walk}$ denote the closed walk in $Q$ obtained by concatenating of the labels of edges in $\Walk$.  The closed walks $\Walk$ and $\seq{\Walk}$ are freely homotopic in $S$; in particular, $\Walk$ is contractible if and only if $\seq{\Walk}$ is contractible.  Moreover, the number of edges in $\seq{\Walk}$ is at most twice the number of edges in $\Walk$.


The universal cover $\Tiling$ of a system of quads $Q$ is isomorphic to a regular tiling of the hyperbolic plane by quadrilaterals meeting at vertices of degree $4g$; see Figure \ref{F:universal-cover}.  The covering map from $\Tiling$ to $Q$ is a graph homomorphism (mapping vertices to vertices and edges to edges), which at the risk of confusing the reader, we also denote $\seq{\cdot}\colon \Tiling \to Q$.  A closed walk  in $Q$ is contractible if and only if it is the projection of a closed walk in $\Tiling$.  Thus, a closed walk $\Walk$ in $G$ is contractible if and only if there is a closed walk $\Lift{\Walk}$ in~$\Tiling$ such that $\seq{\Lift{\Walk}} = \seq{\Walk}$.

\begin{figure}[ht]
\centering
\includegraphics[scale=0.4]{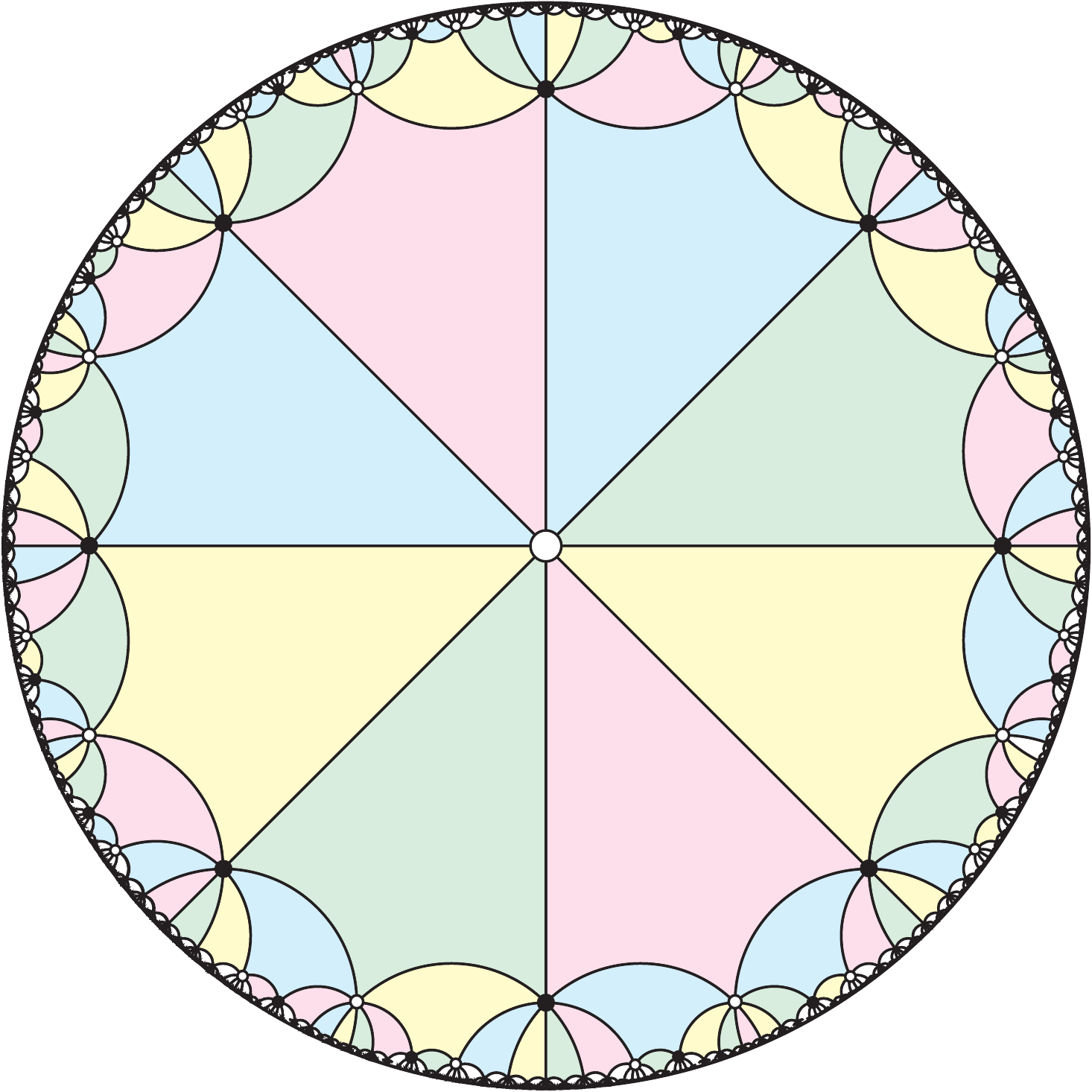}
\caption{The universal cover of a genus-2 system of quads.  (Compare with Figure \ref{F:reduction}.)}
\label{F:universal-cover}
\end{figure}

\subsection{Hyperbolic Geometry}

Our algorithm exploits two well-known geometric properties of regular hyperbolic tilings: the area of any ball in a hyperbolic tiling grows exponentially with its radius~\cite{m-ncfg-68,gh-prges-97}, and the area enclosed by any simple cycle grows at most linearly with the cycle's length~\cite{d-tkzf-12,d-pgtt-87}.  In particular, we rely on the following immediate corollary of these two properties: The distance from any point in the interior of a simple cycle in a regular hyperbolic tiling is at most \emph{logarithmic} in the length of that cycle; see Figure~\ref{F:hyperblob}.  

\begin{figure}[ht]
\centering
\includegraphics[scale=0.4,page=1]{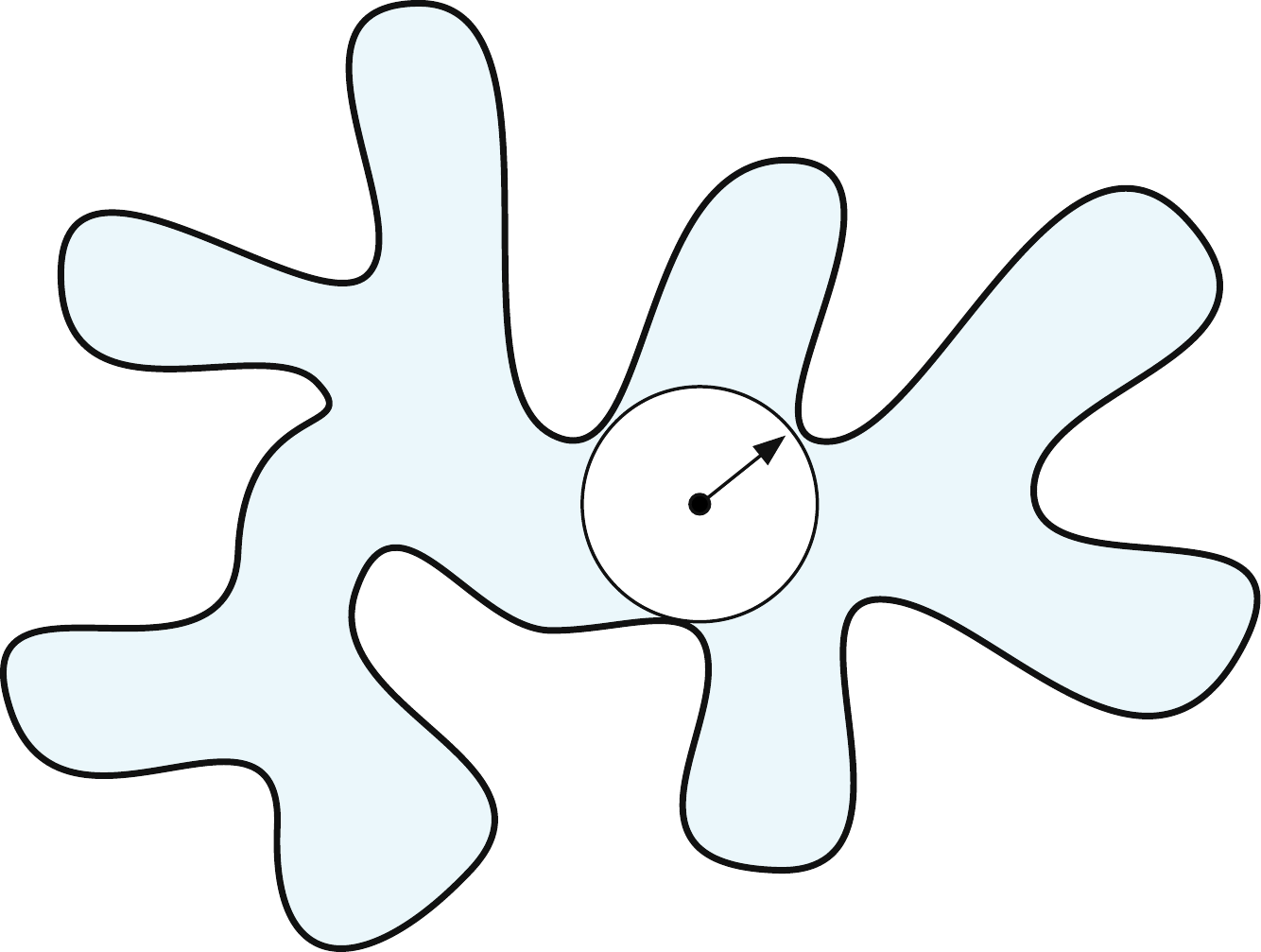}
\caption{Simple closed curves in hyperbolic tilings have logarithmic in-radius.}
\label{F:hyperblob}
\end{figure}

Classical results already imply this \emph{qualitative} behavior for \emph{all} regular hyperbolic tilings; indeed, Chepoi \etal~\cite{cdehv-dcatd-08} have derived tight \emph{asymptotic} bounds.  However, we require \emph{precise} bounds for growth, isoperimetry, and in-radius in the specific tiling $\Tiling$, because the running time of our eventual algorithm depends \emph{exponentially} on the constants in those bounds.

Fix an arbitrary vertex $\Lift{a}$ of $\Tiling$ called the \emph{basepoint}.  For any positive integer $r$, let \EMPH{$n(r)$} denote the number of vertices in $\Tiling$ at distance exactly $r$ from $\Lift{a}$, and let \EMPH{$N(r)$} denote the number of vertices at distance at most $r$ from $\Lift{a}$.

\begin{lemma}
\label{L:exp-growth}
$N(r) \sim \beta \lambda^r$ for some constants $1 < \beta < 2$ and $4g-3 < \lambda < 4g-2$.  In particular, $N(r) \ge \lambda^r$ for all $r\ge 0$.
\end{lemma}

\begin{proof}
Every node at distance $r$ from the basepoint has either zero or one neighbors at distance $r-1$ from the basepoint.  Thus, we can write $n(r) = n_1(r) + n_2(r)$, where each $n_i(r)$ is the number of nodes at distance~$r$ that have $i$ neighbors at distance $r-1$.  We immediately have $n_1(1) = 4g$ and $n_2(1) = 0$.  Floyd and Plotnick \cite{fp-gffge-87} prove that these functions satisfy the following recurrence:
\[
	\begin{bmatrix}
		n_1(r) \\ n_2(r)
	\end{bmatrix}
	~=~
	\begin{bmatrix}
		4g-3 & 4g-4 \\
		1 & 1
	\end{bmatrix}
	\begin{bmatrix}
		n_1(r-1) \\ n_2(r-1)
	\end{bmatrix}
	~=~
	\begin{bmatrix}
		4g-3 & 4g-4 \\
		1 & 1
	\end{bmatrix}^{r-1}
	\begin{bmatrix}
		4g \\ 0
	\end{bmatrix}
\]
(Moran \cite{m-grbht-97} derives similar recurrences for a slightly different growth function.)  Standard techniques now imply that $n(r) = \alpha(\lambda^r -\lambdabar^r)$ for some constant $\alpha$, where $\lambda$ and $\lambdabar = 1/\lambda$ are the roots of the polynomial $z^2 - (4g-2)z + 1$; specifically, we have
\begin{align*}
	\alpha &= \sqrt{g/(g-1)},
	& \lambda &= 2g - 1 + 2\sqrt{g(g-1)},
	& \lambdabar &= 2g - 1 - 2\sqrt{g(g-1)}.
\end{align*}
Straightforward calculation implies $4g-3 < \lambda < 4g-2$.  It follows that $N(r) \sim \beta\lambda^r$, where $1 < \beta = \alpha/(1-\lambdabar)< 2$.

Moreover, a straightforward but tedious induction argument implies that $n_1(r) \ge \lambda^r$ and $n_2(r+1) \ge \lambda^r$ for all $r\ge 1$.  We conclude that $N(r) \ge n_1(r) \ge \lambda^r$ for all $r\ge 1$; the remaining case $r=0$ is trivial.
\end{proof}

\begin{lemma} 
\label{L:dehn-faces}
Any simple cycle of length $L$ in $\Tiling$ has less than $3L/2$ faces of $\Tiling$ in its interior.
\end{lemma}

\begin{proof}
Our proof follows the “spur and bracket” analysis of Erickson and Whittlesey \cite[Section 4.2]{ew-tcsr-13}, which is based in turn on results of Gersten and Short \cite{gs-sctag-90}.

Let $\gamma$ be a nontrivial simple cycle in $\Tiling$, without loss of generality oriented counterclockwise around its interior faces (if any).  The \emph{turn} of any vertex of $\gamma$ is the number of interior tiles incident to~$v$ that locally lie to the left of $\gamma$.  A \emph{(left) bracket} is a subpath of $\gamma$ whose first and last vertices have turn $1$ and whose intermediate vertices have turn $2$.  A discrete curvature argument, ultimately based on Euler's formula, implies that any $\gamma$ has at least four brackets \cite[Corollary~5.1]{gs-sctag-90}\cite[Lemma 4.1]{ew-tcsr-13}.  In particular, $\gamma$ has at least four distinct vertices with turn $1$.

Let $L(\gamma)$ denote the number of edges in $\gamma$, let $A(\gamma)$ denote the number of interior tiles enclosed by $\gamma$, and let $\Phi(\gamma) = 2t_1 + t_2$, where $t_i$ denotes the number of vertices of $\gamma$ with turn~$i$.  Trivially, $\Phi(\gamma) \le 2L(\gamma)$.  

We can reduce the length of $\gamma$ by \emph{sliding} a bracket, as shown in Figure \ref{F:bracket}.  Intuitively, sliding a bracket of length~$\ell$ replaces a subpath of length $\ell+2$ around three sides of a $1\times \ell$ rectangle of tiles with a path of length~$\ell$ around the fourth side of that rectangle.  This replacement removes the turn-$1$ vertices at the ends of the bracket, changes the turns of the vertices in the interior bracket from $2$ to $4g-2$, and decreases the turns of the vertices adjacent to the bracket by~$1$.

\begin{figure}[ht]
\centering
\includegraphics[scale=0.5]{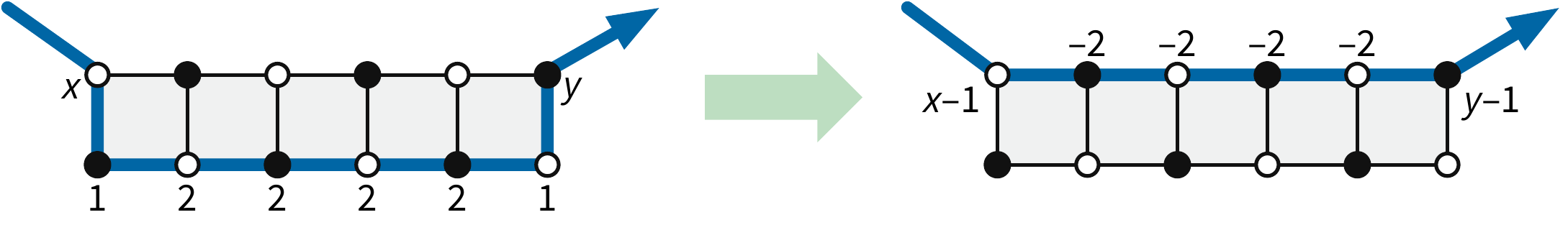}
\caption{Removing a bracket; vertex labels are turns (modulo $4g$).}
\label{F:bracket}
\end{figure}

Sliding a bracket creates a spur (an edge followed by its reversal) if and only if that bracket is adjacent to a vertex with turn $1$, and thus another bracket with length~$1$.  The only simple cycle that contains two consecutive length-$1$ brackets is the boundary of a single tile.  Thus, either $\gamma$ is a single tile boundary, or removing the \emph{shortest} bracket yields a shorter simple cycle $\gamma'$.  In the latter case, we immediately have 
\[
	L(\gamma') = L(\gamma) - 2, \quad
	A(\gamma') = A(\gamma) - \ell, \quad\text{and}\quad
	\Phi(\gamma') \le \Phi(\gamma) - \ell - 1,
\]
where $\ell$ is the length of the bracket.
It follows by induction that $\gamma$ can be reduced to a single tile boundary by a sequence of $L(\gamma)/2 - 2$ bracket slides, and
\[
	A(\gamma) = 1 + \sum_i \ell_i 
	\quad\text{and}\quad
	\Phi(\gamma) \ge 4 + \sum_i (\ell_i + 1)
\]
where $\ell_i$ is the length of the $i$th bracket.  We conclude that $1 + A(\gamma) + L(\gamma)/2 \le \Phi(\gamma) \le 2L(\gamma)$, and therefore
 $A(\gamma) \le 3L(\gamma)/2 - 1$.
\end{proof}

\begin{lemma} 
\label{L:dehn-verts}
Any simple cycle of length $L$ in $\Tiling$ has less than $2L/g$ vertices of $\Tiling$ in its interior.
\end{lemma}

\begin{proof}
The dual of $\Tiling$ is another regular tiling $\Tiling^*$ of the hyperbolic plane by $4g$-gons, meeting $4$ at each vertex.  A classical lemma of Dehn \cite{d-tkzf-12,d-pgtt-87} (see also Lyndon and Schupp \cite{ls-cgt-01}) implies that any closed walk~$\gamma^*$ in the dual tiling $\Tiling^*$ contains either a spur or at least $4g-2$ consecutive edges of some tile.  It follows by induction that $\gamma^*$~encloses at most $L(\gamma^*)/(2g-1)$ tiles in $\Tiling^*$.

Now let $\gamma$ be a simple cycle in the original tiling $\Tiling$.  Offsetting $\gamma$ slightly into its interior yields a closed curve $\gamma^*$ that visits every face in the interior of $\gamma$ at most twice; we can regard $\gamma^*$ as a closed walk in the dual tiling $\Tiling^*$.  The previous lemma implies $L(\gamma^*) < 3L(\gamma)$.  Each vertex of $\Tiling$ in the interior of $\gamma$ is dual to a unique face of~$\Tiling^*$ enclosed by $\gamma^*$.  We conclude that $\gamma$ encloses at most $3L(\gamma)/(2g-1) < 2L/g$ interior vertices.
\end{proof}

The main result of this section now follows from Lemma~\ref{L:dehn-verts} and the exact lower bound in Lemma~\ref{L:exp-growth}.

\begin{lemma}
\label{L:short-radius}
Let $\gamma$ be a simple cycle in $\Tiling$ with length at most $L$.  Every tiling vertex in the interior of $\gamma$ has distance at most $\Radius = \ceil{\log_\lambda (2L/g)}$ from at least one vertex of $\gamma$.
\end{lemma}

\subsection{Triangulation}

Now let $\Walk$ be a closed walk in $\Tiling$ with a distinguished vertex $\Lift{a}$ called its \EMPH{basepoint}.  Following Muller and Schupp~\cite{ms-gtecl-83}, we define a \EMPH{$k$-triangulation} of~$\Walk$ to be a continuous map $\mathcal{T}$ from a \emph{reference} triangulation~$T$ of a simple polygon~$P$ to $\Tiling$ with the following properties:
\begin{itemize}\itemsep0pt
\item
$\mathcal{T}$ maps vertices of $P$ to vertices of $\Walk$.
\item
$\mathcal{T}$ maps a distinguished edge of $P$, called the \EMPH{root edge}, to the basepoint $\Lift{a}$.
\item
$\mathcal{T}$ maps the boundary of $P$ continuously onto the walk $\Walk$; in particular, every edge of $P$ except the root edge is mapped to a single edge of $\Walk$.
\item
Finally, $\mathcal{T}$ maps each diagonal edge in $T$ to a walk of length at most $k$ in $\Tiling$.
\end{itemize}
This map is not necessarily an embedding, even locally.  The diagonal walks may intersect each other or~$\Walk$, even if $\Walk$ is a simple cycle; moreover, some diagonal walks may have length zero if $\Walk$ is not simple.

\begin{figure}[ht]
\centering
\includegraphics[scale=0.4,page=3]{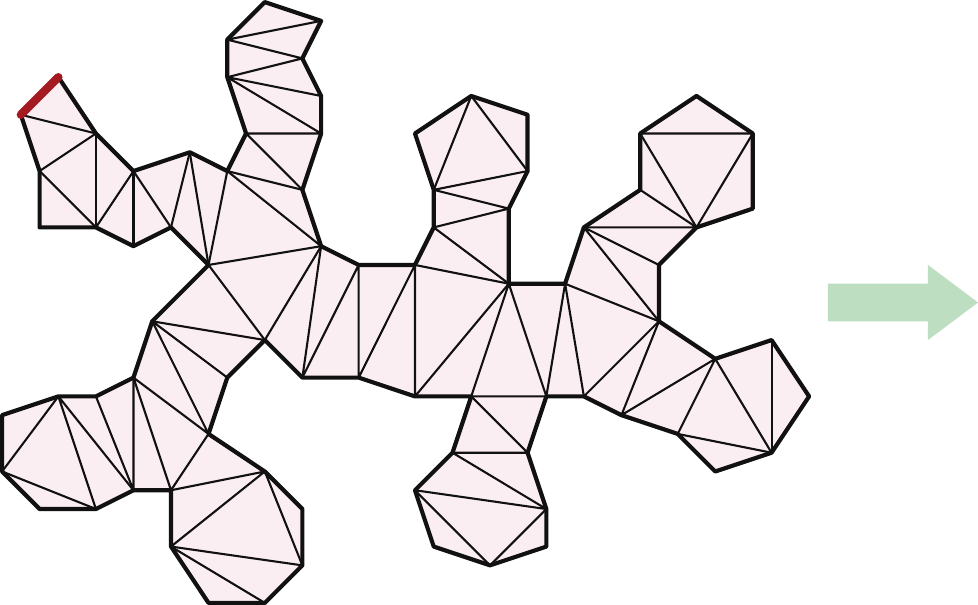}
\includegraphics[scale=0.4,page=2]{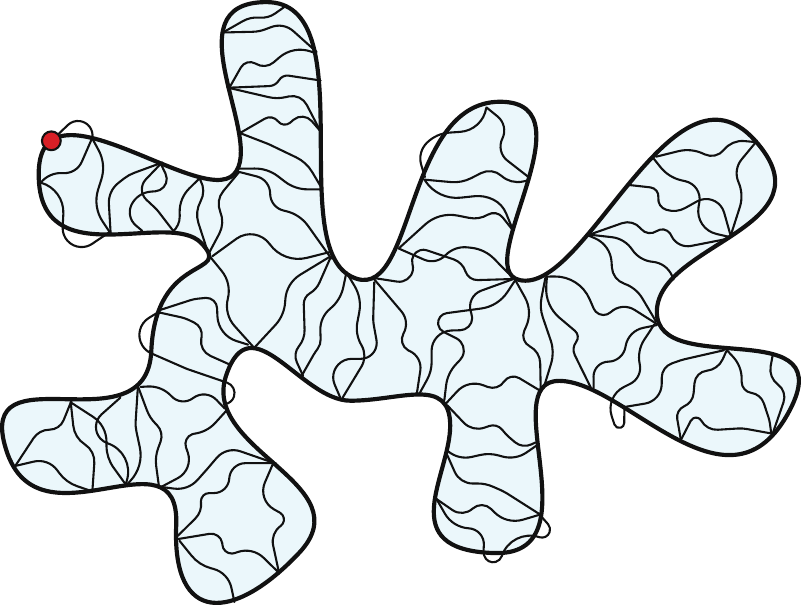}
\caption{Any closed walk in a hyperbolic tiling can be “triangulated” by paths of logarithmic length.}
\end{figure}

\begin{lemma}
\label{L:triangulation-simple}
Every simple cycle of length $L$ in $\Tiling$ has a $(2\Radius+2)$-triangulation, where $\Radius = \ceil{\log_\lambda(2L/g)}$.
\end{lemma}

\begin{proof}
For the moment, let $\gamma$ be a \emph{simple} cycle of length $L$ in $\Tiling$, and let $\Gamma$ denote the closed disk bounded by $\gamma$.  Intuitively, we construct a “Delaunay triangulation” of the vertices of $\gamma$ inside the disk $\Gamma$.

For each tiling vertex $x$ in $\Gamma$, let \EMPH{$nn(x)$} denote its nearest neighbor on $\gamma$, measuring distance by counting edges in $\Tiling$ and breaking ties arbitrarily.  In particular, if $x\in \gamma$, then $nn(x)=x$.  Lemma \ref{L:short-radius} implies that the distance from any vertex $x$ to its nearest neighbor $nn(x)$ is at most $\Radius$.

Arbitrarily triangulate each tile of $\Gamma$ to obtain a triangulation $\Delta$, and let $M$ denote the set of all edges $xy$ in $\Delta$ such that $nn(x)\ne nn(y)$.  (Intuitively, the duals of edges in $M$ form a “discrete medial axis” of~$\Gamma$~\cite{b-cccda-94}.)  For each edge $xy$ in $M$, let \EMPH{$\sigma(x,y)$} denote any shortest path in $\Tiling$ (not in $\Delta$!) from $nn(x)$ to $nn(y)$.  (In particular, if $xy$ is an edge of~$\gamma$, then $\sigma(x,y) = xy$.)  The triangle inequality implies that each shortest path $\sigma(xy)$ has length at most $2\Radius+2$.%
\footnote{A nearly identical argument implies that for any $t>1$, any simple cycle of length $L$ in a regular hyperbolic tiling by $2t$-gons has a $(2\rho + t)$-triangulation, for some appropriate inradius function $\rho = O(\log_\lambda L)$.  Unfortunately, the $+t$ term in the triangulation constant implies a $g^{O(t)}$ factor in the complexity of our context-free grammar, and therefore in the running time of our algorithm.  In particular, using a system of loops would lead to an extra $g^{O(g)}$ factor in the running time. \label{fn:whyquads}}

The shortest paths $\set{\sigma(x,y) \mid xy\in M(\gamma)}$ are the diagonal paths of a $(2\Radius+2)$-triangulation of $\gamma$.  Specifically, for every triangle $xyz$ in $\Delta$ whose vertices have three distinct nearest neighbors, the shortest paths $\sigma(x,y)$, $\sigma(y,z)$, and $\sigma(x,z)$ define a triangle $\triangle(xyz)$ in the $(2\Radius+2)$-triangulation.
\end{proof}

\begin{lemma}
\label{L:triangulation}
Every closed walk of length $L$ in $\Tiling$ has a $(2\Radius+2)$-triangulation, where $\Radius = \ceil{\log_\lambda(2L/g)}$.
\end{lemma}

\begin{proof}
The proof proceeds by induction on $L$.  Let $\Walk$ be any closed walk of length $L$.  If $\Walk$ is a simple cycle, we can defer to the previous lemma, so assume otherwise.  There are two cases to consider.

First, suppose $\Walk$ contains a spur $x\arcto y\arcto x$, and let $\Walk'$ be the walk obtained by removing that spur.  The induction hypothesis implies that $\gamma'$ has a $(2\Radius+2)$-triangulation $\mathcal{T}'$.  There are two subcases to consider; see the top row of Figure \ref{F:extend}.
\begin{itemize}
\item
If $\mathcal{T}$ contains the path $\sigma(xz)$, we can transform $\mathcal{T}$ into a $(2\Radius+2)$-triangulation of $\Walk$ by expanding $\sigma(yz)$ into two triangles $\triangle(yzy)$ and $\triangle(yxy)$.
\item
Otherwise, the triangulation~$\mathcal{T}'$ must contain the path $\sigma(wz)$ between the neighbors of~$x$ on $\Walk'$, either as an edge of $\Walk$ or as a diagonal path.  We can transform $\mathcal{T}'$ into a $(2\Radius+2)$-triangulation of $\Walk$ by replacing the triangle $\triangle(wxz)$ with three triangles $\triangle(wxz)$, $\triangle(wxx)$, and $\triangle(xxy)$.
\end{itemize}
In either case, each new diagonal path is either empty or a duplicate of some path in $\mathcal{T}'$, and thus has length at most $2\Radius+2$.

\begin{figure}[ht]
\centering
\includegraphics[scale=0.4]{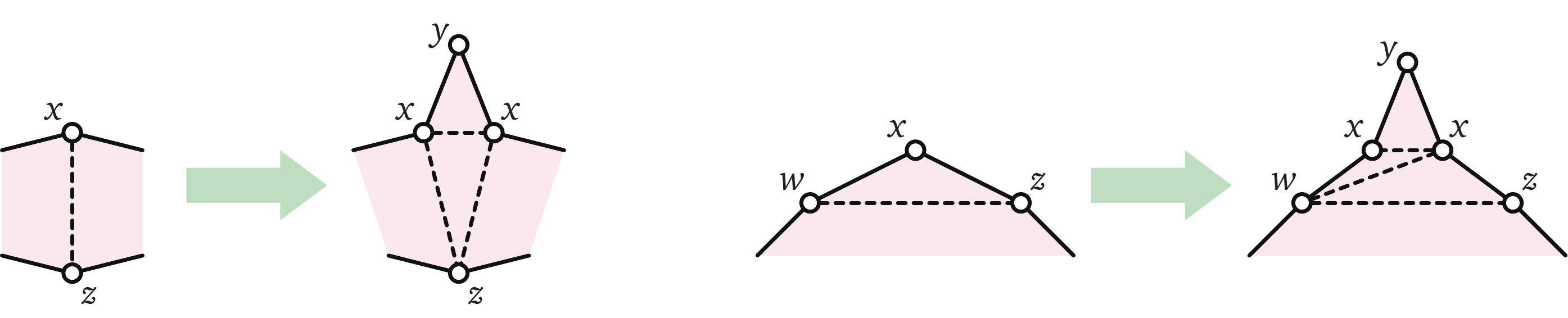}\\[1ex]
\includegraphics[scale=0.4]{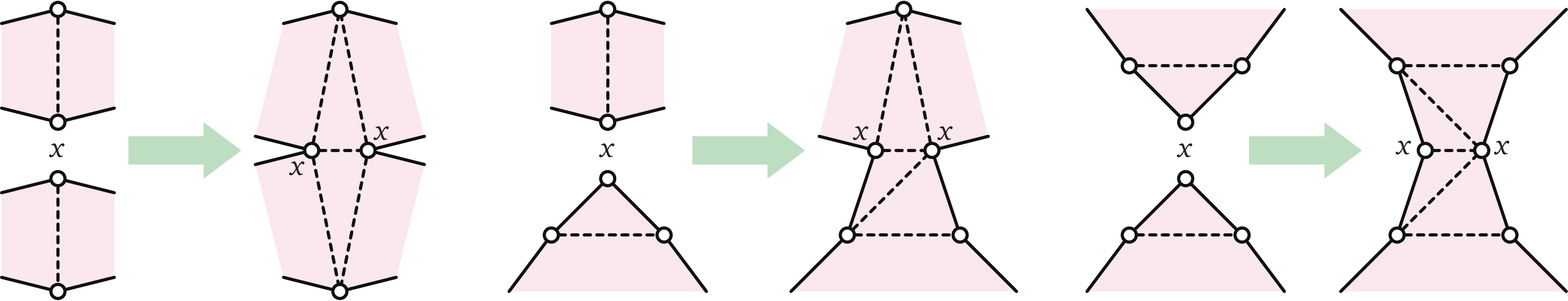}
\caption{Top:~Extending a $k$-triangulation across a spur.  Bottom:~Merging two $k$-triangulations across a repeated vertex.}
\label{F:extend}
\end{figure}

Now suppose $\Walk$ does not contain a spur, but does contain a repeated vertex $x$.  We can partition $\Walk$ into two smaller closed walks $\Walk_1$ and $\Walk_2$ with a common vertex $x$.  The induction hypothesis implies that $\Walk_1$ and~$\Walk_2$ have $(2\Radius+2)$-triangulations $\mathcal{T}_1$ and $\mathcal{T}_2$.  We can merge $\mathcal{T}_1$ and $\mathcal{T}_2$ into a $(2\Radius+2)$-triangulation $\mathcal{T}$ of $\Walk$ by duplicating $x$ and inserting new diagonal paths.  There are three subcases to consider, depending on whether $x$ is the endpoint of a diagonal path in both, only one, or neither of $\mathcal{T}_1$ and $\mathcal{T}_2$; see the bottom row of Figure \ref{F:extend}.  Again, each new diagonal path is either empty or a duplicate of some path in $\mathcal{T}_1$ or $\mathcal{T}_2$, and thus has length at most $2\Radius+2$.

In all cases, we obtain a $(2\Radius+2)$-triangulation of the original walk $\Walk$.
\end{proof}

\subsection{Context-Free Grammar}

Now finally we reach the key component of our algorithm: the construction of a small context-free grammar $\Grammar$ such that the shortest contractible closed walk in $G$ is also the shortest walk whose label is in the language generated by $\Grammar$.  Our construction is essentially the same as that of Muller and Schupp \cite[Theorem 1]{ms-gtecl-83}, but with more careful analysis.

\begin{lemma}
\label{L:CFG}
For any positive integer $L$, there is a context-free grammar $\Grammar$ in Chomsky normal form with $O(g^2L^2)$ non-terminals, such that (1) Every string generated by $\Grammar$ describes a contractible closed walk in~$Q$, and (2) every contractible closed walk in~$Q$ of length at most $L$ is generated by $\Grammar$.
\end{lemma}

\begin{proof}
As in the previous proofs, let $\Radius = \ceil{\log_\lambda(2L/g)}$.

The alphabet~$\Sigma$ of $\Grammar$ consists of all $8g$ darts in the system of quads $Q$.

Call a homotopy class of walks in $Q$ \emph{short} if it contains a walk of length at most $2\Radius+2$.  The non-terminals of $\Grammar$ correspond to short homotopy classes of walks in $Q$.  More formally, fix an arbitrary lift $\Lift{a}$ of $a$ and an arbitrary lift $\Lift{z}$ of $z$.  Each short homotopy class corresponds to a unique ordered pair $(\Lift{s}, \Lift{t})$ of vertices in $\Tiling$, where $\Lift{s} \in \set{\Lift{a}, \Lift{z}}$ and the shortest-path distance from $\Lift{s}$ to $\Lift{t}$ in~$\Tiling$ is at most $2\Radius+2$.  Thus, the number of short homotopy classes, and therefore the number of non-terminals in $\Grammar$, is at most
\[
	2\, N(2\Radius+2)
	~=~ 
	O(\lambda^{2\Radius + 2})
	~\le~
	O(\lambda^{2\log_\lambda(2L/g) + 4})
	~=~
	O((2L/g)^2\, \lambda^4)
	~=~
	O(g^2 L^2)
\]
by Lemma \ref{L:exp-growth} and the definition of $\Radius$.  (Moreover, Lemma \ref{L:exp-growth} implies that the constant hidden in the final $O(\,)$ bound is at most $8$.)

The starting non-terminal of $\Grammar$ is the homotopy class of contractible walks in~$G$ that start and end at~$a$, represented by the vertex pair $(\Lift{a}, \Lift{a})$.

Non-terminal productions in $\Grammar$ correspond to pairs of short homotopy classes whose concatenation is also a short homotopy class.  Specifically, if $\alpha$ and $\beta$ are walks in short homotopy classes $A$ and $B$, and their concatenation $\alpha\cdot\beta$ is in short homotopy class $C$, then $\Grammar$ contains the production $C\to AB$.  There are trivially at most $\abs{N}^2 = O(g^4 L^4)$  non-terminal productions.\footnote{We believe the number of non-terminal  productions is actually $\Theta(g^3 L^3)$; however, this refinement would only improve our final running time to $O(g^5 m^8)$, which doesn't seem worth the trouble.}

There are $8g$ terminal productions $A \to d$ in $\Grammar$, one for (the homotopy class of) each dart $d$ in $Q$.  We deliberately exclude $\e$-productions from $\Grammar$ to avoid generating the empty string.

Every string generated by the nonterminal $(\Lift{s}, \Lift{t})$ is the projection of a non-empty walk from $\Lift{s}$ to $\Lift{t}$ in~$\Tiling$; all such walks are in the same (short) homotopy class.  In particular, every string generated by the starting non-terminal $(\Lift{a}, \Lift{a})$ is a non-empty contractible walk in $Q$ from $a$ to~$a$.

On the other hand, let $\Walk$ be any non-empty closed walk of length at most $L$ in $\Tiling$.  Lemma \ref{L:triangulation} implies that $\gamma$ has a $(2\Radius+2)$-triangulation $\mathcal{T}$.  The dual tree of this triangulation is a parse tree for $\Walk$.  Specifically, the (weak) dual of the reference triangulation $T$ is a binary tree $T^*$.  Direct the dual of the root edge of $T$ into the polygon, and  direct every other edge of $T^*$ away from the root.  For each edge of $T$, label the \emph{head} of its directed dual edge with the homotopy class of its image in $\Tiling$.  With this labeling, $T^*$ is a parse tree for $\seq{\Walk}$ with respect to $\Grammar$.
\end{proof}

\begin{theorem}
\label{Th:short-con-g}
Let $G$ be a directed graph with $m$ non-negatively weighted edges, embedded on the orientable surface of genus $g$ with no boundary.  We can compute the shortest contractible closed walk in $G$ in $O(g^6 m^9)$ time.
\end{theorem}

\begin{proof}
We construct the system of quads $Q$ and label the edges of $G$ in $O(m)$ time.  If necessary, subdivide edges of $G$ so that each edge is labeled with either the empty walk or a single dart in $Q$.  After subdivision, $G$ has at most $2m$ edges, and therefore (because $G$ is symmetric) at most $2m$ vertices.

Then we build a context-free grammar $\Grammar$ that generates only contractible walks in $Q$ and generates all non-empty contractible walks in $Q$ of length at most $2m$, as described in Lemma~\ref{L:CFG}.  This grammar has $N = O(g^2m^2)$ nonterminals and $P = O(g^4 m^4)$ productions.  Corollary \ref{C:ccw-short} implies that the shortest contractible walk in $G$ has length at most $2m$, and therefore has a label of length at most $2m$.  It follows that $\Grammar$ generates the label of the shortest contractible walk in $G$.

Finally, we compute the shortest closed walk in $G$ whose label is generated by~$\Grammar$, using the CFG-shortest-path algorithm of Barrett~\etal~\cite{bjm-flcpp-00}, in $O(NP(n')^3) = O(g^6 m^9)$ time.
\end{proof}

\paragraph{Acknowledgements.} The first author would like thank Tillmann Miltzow for asking an annoying question that led to this work, and to apologize for still being unable to answer it.

\newpage	
\bibliographystyle{newuser}
\bibliography{topology,compgeom,optimization,algorithms,jeffe,misc}

\end{document}